\documentclass[twocolumn]{autart2}    
\usepackage{amsmath,bm}
\usepackage{abraces}
\usepackage{times}
\usepackage{enumitem}
\usepackage{color}
\usepackage{multirow}
\usepackage{mathrsfs}
\usepackage{tcolorbox}
\usepackage{subfigure}
\usepackage{stfloats}
\usepackage{lineno,hyperref}
\hypersetup{
    colorlinks=true,
    linkcolor=blue,
    filecolor=magenta,      
    urlcolor=magenta
    }
\usepackage{dsfont}
\usepackage{amssymb}
\usepackage{setspace}
\usepackage{graphicx}
\usepackage{float}
\newtheorem{remark}{Remark}
\newtheorem{lemma}{Lemma}
\newtheorem{definition}{Definition}

\newtheorem{assumption}{Assumption}

\newtheorem{proposition}{Proposition}

\newtheorem{proof}{Proof}

\def\begcen{\begin{center}}
\def\endcen{\end{center}}

\newcommand{\col}{\mbox{col}}

\def\diag{\mbox{diag}}

\def\ob#1{{\aoverbrace[L1R]{#1}}}
\def\bfy{{\bf y}}

\def\bby{\mathbb{Y}}
\def\bq{\mathbb{Q}}

\def\calb{\mathcal{B}}
\def\caln{\mathcal{N}}
\def\cale{\mathcal{E}}
\def\cali{\mathcal{I}}
\def\calm{{\mathcal M}}

\def\hal{\frac{1}{2}}

\def\calh{{\mathcal{H}}}

\def\L2e{{\mathcal L}_{2e}}

\def\rea{\mathbb{R}}

\def\diag{\mbox{diag}}
\def\adj{\mbox{adj}}
\def\et{\epsilon_t}

\def\l2{{\mathcal L}_2}
\def\l2e{{\cal L}_{2e}}

\def\rea{\mathbb{R}}

\def\diag{\mbox{diag}}
\def\skew{\mbox{skew}}
\def\tr{\mbox{tr}}

\def\begequarr{\begin{eqnarray}}
\def\endequarr{\end{eqnarray}}
\def\begequarrs{\begin{eqnarray*}}
\def\endequarrs{\end{eqnarray*}}
\def\begarr{\begin{array}}
\def\endarr{\end{array}}
\def\begequ{\begin{equation}}
\def\endequ{\end{equation}}

\def\begdes{\begin{description}}
\def\enddes{\end{description}}
\def\begenu{\begin{enumerate}}
\def\begite{\begin{itemize}}
\def\endite{\end{itemize}}
\def\endenu{\end{enumerate}}

\def\lef[{\left[\begin{array}}
\def\rig]{\end{array}\right]}
\def\qed{\hfill$\blacksquare$}
\def\begcen{\begin{center}}
\def\endcen{\end{center}}
\def\begrem{\begin{remark}\rm}
\def\endrem{\end{remark}}

\def\begmat#1{\begin{bmatrix}#1\end{bmatrix}}

\usepackage{color}

\begin{document}

\begin{frontmatter}

\title{Globally Convergent Visual-Feature Range Estimation with Biased Inertial Measurements\thanksref{footnoteinfo}} 

\thanks[footnoteinfo]{This paper was not presented at any IFAC
meeting. This paper is supported by the Australian Research Council. 
}

\author[ACFR]{Bowen Yi}\ead{\rm bowen.yi@sydney.edu.au},
\author[DJI]{Chi Jin}, 
\author[ACFR]{Ian R. Manchester}


\address[ACFR]{Australian Centre for Field Robotics, Sydney Institute for Robotics and Intelligent Systems, The University of Sydney, NSW 2006, Australia}
\address[DJI]{DJI Innovation Inc. Shenzhen 518057, China}

\begin{keyword}                        
Observers Design; Nonlinear Systems; Range Estimation; Robotics             
\end{keyword}                             

\begin{abstract}
The design of a globally convergent position observer for feature points from visual information is a challenging problem, especially for the case with only inertial measurements and without assumptions of uniform observability, which remained open for a long time. We give a solution to the problem in this paper assuming that only the bearing of a feature point, and biased linear acceleration and rotational velocity of a robot --- all in the body-fixed frame --- are available. Further, in contrast to existing related results, we do not need the value of the gravitational constant either. The proposed approach builds upon the parameter estimation-based observer recently developed in (Ortega et al., {\em Syst. Control Lett.}, vol. 85, 2015) and its extension to matrix Lie groups in our previous work. Conditions on the robot trajectory under which the observer converges are given, and these are strictly weaker than the standard persistency of excitation and uniform complete observability conditions. Finally, as an illustration, we apply the proposed design to the visual inertial navigation problem. 
\vspace{-.5cm}
\end{abstract}
\end{frontmatter}

\section{Introduction}
\label{sec1}

Determination of the position of feature points from visual measurements is a classical problem in several fields, including machine vision, robotics and control \cite{ASTetal02,BERetal,DANetal,DELetal,LOUetal,ZHAZEL}. However, its mathematical model does not fall into the canonical forms that have been comprehensively studied in the nonlinear observer community \cite{BES}, thus being of theoretical interest. With a single monocular camera, estimating the position of a feature point is equivalent to estimating its range or depth. There are two general classes of methodologies: batch methods \cite{HANKAN} and observer design \cite{DELetal,HUetal,KARAST}. In this paper, we focus on the latter for its simpler on-line computations.

In \cite{JANGHO}, the authors adopt a change of coordinate to the inverse of depth, which has since become a popular formulation, and propose an identifier based observer. Its main drawback is high-gain injection thus yielding sensitivity to measurement noise. Thereafter, many nonlinear observers were proposed using different constructive tools. A less complicated observer design based on Lyapunov analysis is proposed in \cite{DIXetal}, which guarantees global asymptotic stability (GAS) under an instantaneous observability assumption on the robot trajectory. An extension to paracatadioptric camera can be found in \cite{HUetal}, where a sliding-mode observer is proposed for both affine and non-affine systems. In \cite{KARAST}, a reduced-order observer is proposed based on the widely used immersion and invariance (I\&I) technique \cite{ASTetal}, achieving semi-global convergence under the same condition in \cite{DIXetal}. To address visual servoing, a simple feature depth observer is introduced in \cite{DELetal}, which builds upon a well-known lemma from adaptive control. This design requires a persistency of excitation (PE) condition --- weaker than the condition in \cite{DIXetal,KARAST} --- however, with only local convergence guaranteed. More recently it is shown that by selecting alternative outputs we may obtain a linear time-varying (LTV) model \cite{HAMSAM,LOUetal}, for which Kalman-Bucy or Riccati observer is applicable to achieve global stability under some PE conditions. In the last decade, practical considerations related to this problem have also received significant attention \cite{SASetal,SPIetal,TAHetal}.

In all the above works, it is assumed a measurement of linear velocity available in observer design. However, a practically important scenario is that the robot is equipped with inertial measurement unit (IMU), which provides measurements of linear acceleration rather than velocity. In spite of intensive research efforts, we are unaware of any observer for the extension to this case. The main challenge relies on that the unknown attitude --- living in the special orthogonal group --- appears in the dynamics of the body-fixed velocity. 

In this paper, we give the first constructive solution to the problem, presenting a novel globally exponentially convergent position observer for a single feature point. The constructive tool we adopt is the recently-introduced parameter estimation-based observer (PEBO) in \cite{ORTetalscl} and its extension to matrix Lie groups in our previous work \cite{YIetalCDC} on the simultaneous localization and mapping (SLAM) problem. The central idea in PEBO is to translate the estimation of a (time-varying) system state into one of estimating constant parameters. Another merit is that, unlike all existing solutions, the proposed observer neither requires instantaneous observability \cite{KARAST}, uniform complete observability \cite{LOUetal}, nor PE conditions \cite{DELetal}, and its convergence is achieved under a much weaker interval excitation (IE) condition. 

Our second contribution is applying the proposed scheme to visual inertial navigation, arising for in Global Positioning System (GPS)-denied environment, which is concerned with fusing the information from IMUs and cameras, using the inertial coordinates of some feature points in a known map. Many types of Kalman filters have became the \emph{de facto} standard algorithms in industrial applications \cite{MOUROU}, which, however, are intrinsically local. In order to enlarge domain of attraction, many researchers from the nonlinear control community have made important contributions in recent years \cite{BERetal,WANetal}. In this paper, we provide a simple, almost globally asymptotically convergent solution only requiring the trajectory being IE. The relaxed excitation condition is especially important for the case with unknown gravitational constant and sensor bias, for which it is quite strict to impose the uniformity of excitation with respect to time. 


{\em Notation.} We use generally $\et$ to represent exponentially decaying terms with proper dimensions. $0_n \in \rea^n$ and $0_{n\times m}\in \rea^{n\times m}$ denote the zero column vector of dimension $n$ and the zero matrix of dimension $n\times m$, respectively. We use $p$ to represent the differential operator $p:= {d\over dt}[\cdot]$, and $|\cdot|$ as the Euclidean norm of a vector. Given a square matrix $A\in \rea^{n\times n}$, the Frobenius norm is defined as $\|A\| = \sqrt{\tr(A^\top A)}$, and $\adj(A)$ and $\det(A)$ denote its adjugate matrix and determinant, respectively. We use $SO(3)=\{R\in \rea^{3\times3}|R^\top R = I_3, ~ \det(R) =1\}$ to represent the special orthogonal group, and ${\mathfrak {so}}(3)$ is its Lie algebra. Given a variable $R\in SO(3)$, we use $|R|_I$ to represent the normalized distance on $SO(3)$ with $|R|_I^2:= {1\over4}\tr(I_3 - R) $. For any $x\in \rea^3/\{0\}$, its projector is defined as
$
\Pi_x :=I_3 -{1\over|x|^2}xx^\top.
$
Given $a \in \rea^3$, we define the operator $(\cdot)_\times$ as
$
a_\times := \begmat{ 0 & - a_3 & a_2 \\  a_3 & 0 & -a_1 \\ -a_2 & a_1 & 0 } \in {\mathfrak {so}}(3) .
$


\section{Model and Problem Formulation}
\label{sec2}


\begin{figure}[!htb]
    \centering
    \includegraphics[width = 0.7\linewidth]{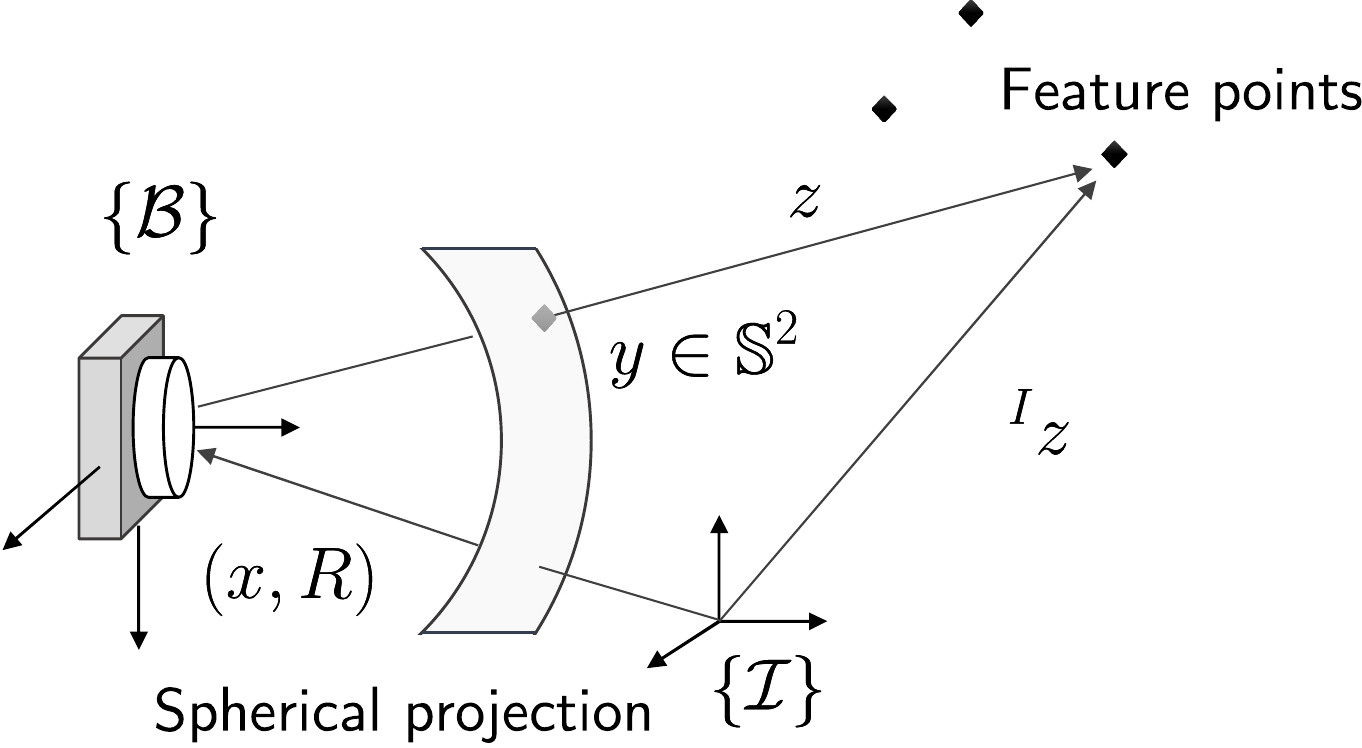}
    \caption{Coordinate systems of a robot observing fixed points}
    \label{fig:coordinate}
\end{figure}

A single monocular camera provides 2D images of its environment, and we focus the 3D position estimation of \emph{point} features extracted from images in the view of a camera on the mobile robot, which is equipped with IMUs. To be precise, in the body-fixed frame $\{\calb\}$ the observed feature point is denoted by $z\in \rea^3$, satisfying
$
	z= R^\top ({}^I z - x)
$
with unknown position $x\in \rea^3$ of the robot in the inertial frame $\{\cali\}$, and $R\in SO(3)$ is the attitude of $\{\calb\}$ with respect to $\{\cali\}$. Here, the constant vector ${}^I z \in \rea^3$ represents the feature position in $\{\cali\}$. The dynamics of $z$ is given by
\begin{equation}
\label{dot_z}
\dot z = - \Omega_\times z - v,
\end{equation}
in which $\Omega \in \rea^3$ is the rotational velocity, and $v\in \rea^3$ is the linear velocity, both in the body-fixed frame $\{\calb\}$; see Fig. \ref{fig:coordinate}. The kinematics of the robot is given by
\begin{equation}
\label{dot_xR}
\begin{aligned}
	\dot x  = {}^I v, \quad
	\dot R  = R \Omega_\times
\end{aligned}
\end{equation}
in which ${}^I v$ is the velocity in $\{\cali\}$, {\em i.e.}
$
{}^I v := Rv.
$
We have
\begin{equation}
\label{dot_v}
\dot v= -\Omega_\times v + a + b_a + R^\top g,
\end{equation}
in which $a\in \rea^3$ is the ``apparent acceleration'' representing all non-gravitational forces in $\{\calb\}$ measured by IMUs, $g:= [0,0,\mathsf{g}]^\top$ is the gravity vector in $\{\cali\}$, and $b_a\in \rea^3 $ is the constant sensor bias. In this paper, we consider the \emph{spherical projection} model of the camera. In this case, the output is the bearing of the feature point in the body-fixed frame $\{\calb\}$, {\em i.e}.
\begin{equation}
\label{y}
y = {z \over |z|} \in \mathbb{S}^2,
\end{equation}
for non-zero $z$. We make the following assumption.

\begin{assumption}
\label{ass:1} The origin of the robot never coincides with the feature point, {\em i.e.} $|z(t)| \neq 0$ for all $t\ge 0$. Besides, $\Omega$ and $a$ guarantee the systems state bounded over time.
\end{assumption}

The proposed approach is applicable to other image-based estimation problems with different projection models. 
Now let us recall some definitions and a technical lemma as follows. Note that the IE condition is called ``exciting over a finite time interval'' in \cite[pp. 108]{TAO}. The swapping lemma for stable filters is useful for dealing with unavailable ``dirty derivatives'' to generate linear regressors.

\begin{definition}\label{def1}\rm
Given a bounded signal $\phi:\rea_+ \to \rea^n$, it is 1) PE if
$
     \int_{t}^{t+T} \phi(s)\phi^\top(s) ds \succeq \delta I_n, ~ \forall t\ge 0
$
    for some $T>0,\delta >0$; 2) IE if there exist $t_0 , t_c \ge 0$ such that
$
    \int_{t_0}^{t_0+t_c} \phi(s)\phi^\top(s) ds \succeq \delta I_n
$
    for some $\delta >0$.
\end{definition}


\begin{lemma}
\label{lem:swapping}
\rm ({\em Swapping lemma} \cite{SASBOD}) For $C^1$-differential signals $x,y: \rea_{\ge 0} \to \rea$, the following holds
\begin{equation}
\label{swap}
{\alpha \over p+ \alpha}[xy]
=
y {\alpha \over p + \alpha}[x] - {1\over p+\alpha} \left[\dot y {\alpha\over p+\alpha}[x] \right]
\end{equation}
for any $\alpha>0$ with $p:={d/dt}$.
\end{lemma}


We are interested in the position estimation problem below.

\begin{prob}
\rm
Consider the system \eqref{dot_z}-\eqref{dot_v} with $\Omega$ and $y$ available.
\begin{itemize}[topsep=0pt,parsep=0pt,partopsep=0pt]
\item[\bf P1] If the velocity $v$ is measurable, design an observer
$
	\dot \zeta  = N(\zeta,\Omega,v,y)
$,
$
	\hat z  = H(\zeta, \Omega,v,y)
$;
or
\item[\bf P2] If only $a$ is measurable with $b_a$ and $g$ unknown, design an observer
$
	\dot \zeta  = N(\zeta,\Omega,a,y),~
	(\hat z,\hat v)  = H(\zeta, \Omega,a,y)
$
\end{itemize}
with the observer state $\zeta \in \rea^{n_\zeta}$ and the mappings $N$ and $H$ of proper dimensions, such that
$
\lim_{t\to\infty} |\hat z(t)-z(t)|=0
$
for a class of robot trajectories.

\end{prob}


%
\section{Position Observer with Velocity Information}
\label{sec3}
%

To start with, in this section we give a solution to {\bf P1} as a motivating design, and in the next section, we will focus on {\bf P2} without the information of $v$. 

\subsection{Generation of Linear Regression Models}
\label{sec31}

For {\bf P1} we show that position estimation is equivalent to identification of a constant scalar parameter. In the following we present a new parameterization to the range, and show how to generate a linear regression equation (LRE).

\begin{proposition}
\label{prop:1}\rm
Consider the dynamics \eqref{dot_z} and the dynamic extension
\begin{equation}
\label{dyn_ext:r}
\dot \xi = - y^\top v
\end{equation}
with arbitrary $\xi(0)\in \rea$. Then, the position $z$ satisfies
\begin{equation}
\label{id:z}
z(t) = [\xi(t) + \theta]y, \quad \forall t \ge 0,
\end{equation}
in which the unknown constant $\theta \in \rea$ verifies the LRE
\begin{equation}
\label{lre:r}
y_{\tt R} = \phi\theta + \et
\end{equation}
with an exponentially decaying term $\et$, measurable signals
\begin{equation}
\label{phi:r}
\begin{aligned}
\phi & ~:= ~G_1[y] + \alpha G_2[\Omega_\times y]
\\
y_{\tt R} & ~:=~ -  G_2[\alpha\Pi_y v + y^\top v \phi]  - \phi \xi,
\end{aligned}
\end{equation}
the filters
$
G_1[\cdot] := {\alpha p \over p+\alpha}[\cdot]
$
,
$
G_2[\cdot] := {1 \over p+ \alpha}[\cdot],
$
and $\alpha>0$.
\end{proposition}
\begin{proof}\rm
From $|z|^2 = z^\top z$ and \eqref{dot_z}, we have
$
|z| \dot{\ob{|z|}}= z^\top (-\Omega_\times z- v) = z^\top v,
$
and thus the bearing of $z$ satisfies 
$
{d\over dt}|z| = -y^\top v.
$
Note that the time derivative of state in the above coordinate is measurable. Thus, the dynamic extension \eqref{dyn_ext:r} guarantees
$
|z(t)| - \xi(t)= \theta, ~ \forall t\ge 0,
$
with $\xi$ an available signal, in which the constant scalar parameter $\theta$ is defined as $\theta:=  |z(0)| - \xi(0)$. Hence, we have verified \eqref{id:z}.

In the remainder, we show how to derive the LRE \eqref{lre:r}. By calculating the time derivatives of \eqref{y}, we obtain
$$
\begin{aligned}
	\dot y & = {1\over |z|^2}
	\Big[ (-\Omega_\times z - v)|z| + zy^\top v\Big]
	\\
	& = - \Omega_\times y - {1\over |z|} v + {1\over|z|}yy^\top v
	 = - \Omega_\times y - {1\over |z|} \Pi_y v,
\end{aligned}
$$
equivalently,
\begin{equation}
\label{id:r}
r(\dot y + \Omega_\times y) = -\Pi_y v,
\end{equation}
where we have defined the variable
$
    r = |z|
$
for convenience. Noting the unavailability of the derivative $\dot y$, we thus apply the filter ${\alpha \over p + \alpha}$ to both sides of \eqref{id:r} with $\alpha >0$. For the first term, we utilize the swapping lemma for filters --- see Lemma \ref{lem:swapping} --- and then obtain
$$
\begin{aligned}
{\alpha \over p+\alpha}[r\dot y]
& ~=~ r{\alpha p \over p+\alpha}[y] - {1\over p +\alpha}
\left[ \dot r {\alpha \over p+\alpha}[\dot y]\right] + \et
\\
& ~=~
r{\alpha p \over p+\alpha}[y] + {1\over p +\alpha}
\left[ y^\top v {\alpha \over p+\alpha}[\dot y]\right]+\et,
\end{aligned}
$$
in which the exponentially decaying term $\et$ arises from the initial conditions of filters. To be precise, the state-space model of the filter ${\alpha \over\alpha + p}[\cdot]$ is given by $\dot\bfy = - \alpha \bfy + \alpha u$ with input $u$ and output $\bfy$. For non-zero $\bfy(0)$, it yields an exponentially decaying term $e^{-\alpha t}\bfy(0)$. Therefore, we have
$$
\begin{aligned}
r{\alpha p \over p+\alpha}[y] + {1\over p +\alpha}
\left[ y^\top v {\alpha \over p+\alpha}[\dot y]\right]
+
{\alpha \over p+\alpha}[r\Omega_\times y]
\\
=
-{\alpha \over p+\alpha}[\Pi_y v] + \et,
\end{aligned}
$$
then
\begin{equation}
\label{filter_r}
\begin{aligned}
r \phi
+ G_2[y^\top v \phi]
= - \alpha G_2[\Pi_y v] + \et.
\end{aligned}
\end{equation}
Substituting $r =\xi+\theta$ into the above equation, we have the LRE \eqref{lre:r}. 
\qed
\end{proof}

Thanks to the algebraic relation \eqref{id:z}, we have translated, via designing the dynamic extension \eqref{dyn_ext:r}, the estimation of the position $z$ into the on-line consistent identification of $\theta$ from the LRE \eqref{lre:r}. Note that $\theta$ is a scalar constant, and $\phi$ is a column vector, and thus the least squares problem is solvable if $\phi(t_\star)$ is non-zero for some $t_\star>0$. We make the following assumption.
\begin{assumption}
\label{ass:a1}
There exists $t_\star >0$ such that 
$
 |\phi(t_\star)| \neq 0.
$
\end{assumption}
%
%

%

%
\subsection{Position Observer Design}
\label{sec:3b}
Based on the LRE \eqref{lre:r}, we design in this section two globally exponentially convergent position observers to address {\bf P1}, {\em i.e.}, a gradient observer and a PEBO. In the former a PE condition is imposed to guarantee convergence; in contrast for the latter, we relax the requirement significantly.


\begin{proposition}\rm
\label{prop:2}({\em Gradient position observer}) Consider \eqref{dot_z} with $v$ and $y$ available. If $\phi^\top$ in \eqref{phi:r} is PE, then the observer
\begin{equation}
\label{grad_obs}
\begin{aligned}
	\dot{\hat r} & = - y^\top v - \gamma\phi^\top \Big(\phi \hat r + G_2[y^\top v \phi] + \alpha G_2[\Pi_y v]\Big)
	\\
	\hat z & = \hat ry
\end{aligned}
\end{equation}
with $\gamma>0$, provides
a globally exponentially convergent estimate to the position $z$, {\em i.e.},
\begin{equation}
\label{converge:z}
	\lim_{t\to\infty} |\hat z(t) - z(t)| =0 \quad \mbox{(exp.)}.
\end{equation}
\end{proposition}
\begin{proof}\rm
See Appendix.
\qed
\end{proof}

The methodology of gradient observers may date back to \cite{SHI} for general nonlinear models, which has been extended to several applications, {\em e.g.}, electric motors \cite{ORTYI} and pose estimation \cite{LAGetal}. Its key step is to obtain a regression model of unknown states, and PE conditions are imposed to achieve uniform convergence. In many cases, it may not be guaranteed for given robot trajectories. To address this, we design a PEBO under a weaker excitation requirement.

\begin{proposition}\rm
\label{prop:3} ({\em Position PEBO}) Considering the dynamics \eqref{dot_z} with $v$ and $y$ available, the observer consisting of \eqref{dyn_ext:r} and
\begin{equation}
\label{pebo:1}
\begin{aligned}
	\dot \zeta & ~=~ \phi^\top y_{\tt R} - \phi^\top\phi \zeta \\
	\dot \omega & ~=~ - \phi^\top \phi \omega, \quad \omega(0)=1
	\\
	\dot {\hat\theta } &~ =~ \gamma [ (\zeta - \omega\zeta_0) - (1-\omega) \hat \theta ] \\
	\hat z &~=~ (\xi+\hat\theta)y
\end{aligned}
\end{equation}
with $\zeta(0)= \zeta_0$, $y_{\tt R}, \phi$ defined in \eqref{phi:r}, and the filters $G_1[\cdot]$ and $G_2[\cdot]$ starting from zero initial conditions, guarantees the global exponential convergence \eqref{converge:z} if Assumption \ref{ass:a1} holds.
\end{proposition}
\begin{proof}\rm
From Proposition \ref{prop:1}, we have
$
r \equiv \xi + \theta.
$
We need to verify $\hat\theta \to \theta$ as $t\to\infty$. The following analysis is motivated by the proof in our previous work \cite[Proposition 1]{YIetalCDC}, where an LTV filter is designed to generate PE regressors from the ones only satisfying IE; see \cite{WANORTetal,BOBetal} for extensions. It yields
$$
\begin{aligned}
	{d\over dt}(\zeta - \theta) =~\phi^\top y_{\tt R} - \phi^\top \phi \zeta
	=
	- \phi^\top \phi (\zeta - \theta),
\end{aligned}
$$
in which we have used \eqref{lre:r}, and the fact that $\theta$ is a constant. It is underlined that, by setting zero initial conditions of the filters $G_1[\cdot]$ and $G_2[\cdot]$, the decaying term $\et$ disappears in \eqref{lre:r}; see Remark \ref{rem:zero} on how to implement it. It is easy to get
$
\zeta - \theta = \omega(\zeta_0 - \theta)
$
with $\omega(t) = \exp(-\int_0^t \phi^\top(s) \phi(s) ds)$, then yielding a new linear regression model
\begin{equation}
\label{lre:new1}
 \zeta - \omega \zeta_0 = (1-\omega)\theta.
\end{equation}
Now, let us show that the new regression $(1-\omega)$ is PE if Assumption \ref{ass:a1} holds. From the continuity of the signal $\phi$, Assumption \ref{ass:a1} implies the existence of a sufficiently small parameter $\epsilon>0$ such that
$
\int_{t_\star}^{t_\star +\epsilon} \phi^\top(s)\phi(s)ds >0.
$
Hence, we conclude that $\phi^\top$ is IE. The solution of $\omega$ satisfies for $t\in [0,t_\star+\epsilon]$, $1-\omega(t) \ge 0$; and for $t> t_\star + \epsilon$
$$
\begin{aligned}
1 - \omega(t) & = 1- \exp\left(
-\int_0^t \phi^\top(s) \phi(s) ds
\right)
\\
& \ge 1 -  \exp\left(
-\int_{t_\star}^{t_\star + \epsilon} \phi^\top(s) \phi(s) ds
\right)
> \delta_0,
\end{aligned}
$$
for some constant $\delta_0>0$. Hence, the regression $(1-\omega)$ in \eqref{lre:new1} is non-negative and PE. We define the parameter estimation error $\tilde \theta:= \hat\theta - \theta$, the dynamics of which is
\begin{equation}
\label{t_theta}
\dot{\tilde \theta} = - \gamma (1-\omega) \tilde\theta, \quad \gamma >1.
\end{equation}
Using the Cauchy–Schwarz inequality for integrals, we conclude that $(1-\omega)^{1\over 2}$ is PE. The global exponential stability of the LTV dynamics \eqref{t_theta} is established \cite[Thm 2.5.1]{SASBOD}. Then, we have $\lim_{t\to\infty} |\xi(t) + \hat\theta(t) - r(t)| =0 $ exponentially. Invoking state boundedness, we complete the proof.\qed
\end{proof}

\begin{remark}\rm
In the new regressor \eqref{lre:new1}, the signal $(1-\omega)\in[0,1)$ --- intuitively, this value characterizes how ``strong'' the excited signal is --- as a result limiting the convergence speed of the observer. A possible way to overcome this issue is to mix \eqref{lre:new1} with the original regressor \eqref{lre:r}, and then obtaining a new LRE
$
 {y}_{\tt R}' = {\phi}'\theta,
$
with ${y}_{\tt R}':=\zeta- \omega\zeta_0 + k_p\phi^\top y_{\tt R}$ and the new regression $\phi':= (1-\omega) + k_p |\phi|^2$, which is not necessarily less than one. Here, the gain $k_p>0$ plays the role to make a tradeoff on the trust of historical and current information. The last equation in the observer \eqref{pebo:1} may be modified accordingly to accelerate convergence speed. 
\end{remark}

\begin{remark}\rm
\label{rem:zero}
We underline that the third equation in the observer \eqref{pebo:1}, in which we have used both the PE condition of $(1-\omega)$ as well as its positiveness after $t_\star + \epsilon$, is not a gradient flow. The gradient flow is given by
$
	{d\over dt}{\hat\theta }  = \gamma (1-\omega)[ (\zeta - \omega\zeta_0) - (1-\omega) \hat \theta ],
$
which also provides a globally exponentially convergent estimate to $\theta$ under the IE assumption. We refer the reader to \cite{SASBOD} for on-line gradient descent algorithms. Besides, in Proposition \ref{prop:3} it is necessary to carefully select the initial conditions of filters $G_1[\cdot]$ and $G_2[\cdot]$ to deal with the term $\et$ for the IE case. It can be implemented as the state space models
$
\dot {\bf y}  = - \alpha {\bf y} + {\bf u}
$
with ${\bf y}(0)= 0$ for $G_2[\cdot]$, and $\dot {\bf x} = - \alpha {\bf x}+ \alpha^2{\bf u}, ~ {\bf y} = -{\bf x} + \alpha {\bf u}$ with ${\bf x}(0) = \alpha {\bf u}(0)$ for $G_1[\cdot]$, in which ${\bf x},~{\bf u}$ and ${\bf y}$ denote the internal state, input and output of the filters, respectively.
\end{remark}


\section{Position-Velocity Observer with Biased Acceleration Measurement}
\label{sec4}

In this section, we present the main result of the paper, {\em i.e.}, a solution to {\bf P2}, for which we design a position-velocity observer with availability of only the acceleration. 

\subsection{Generation of Linear Regression Models}
\label{sec:4a}

%



We are interested in the estimation of position $z$, linear velocity $v$ and the bias $b_a$, which are collected in the vector
$$
X:= \col(z, v, b_a) \in \rea^9.
$$

\begin{proposition}
\rm \label{prop:lre1}
Consider the dynamics \eqref{dot_z}-\eqref{dot_v}, and design the following dynamic extension
\begin{equation}
\label{dyn_ext:2}
\begin{aligned}
\dot Q & = Q \Omega_\times
\\
\dot\xi & =   A (y,\Omega, Q)  \xi + B(a), \quad
\xi(0)= 0_{10}
\\
\dot \Psi & =  A (y,\Omega, Q) \Psi, \quad \Psi(0) = I_{10}
\end{aligned}
\end{equation}
with $Q(0) \in SO(3)$, and 
\begequ
\label{AB:prop4}
A:= \begmat{
0 & - y^\top & 0_{3}^\top & 0_3^\top
\\
0_{3} & -\Omega_\times & I_3 & Q^\top
\\
0_6 & \ldots & \ldots & 0_{6\times 3}
}
\quad
B:= \begmat{0\\ a \\ 0_{6}}.
\endequ
Then, the state $X$ satisfies the algebraic equation
\begin{equation}
\label{algebraic:4.1}
    X = T(y)[\xi + \Psi\theta]
\end{equation}
with 
$
T(y):= \diag\{y, [I_6, 0_{6\times 3}] \},
$
in which the unknown constant vector $\theta$ admits the LRE
$
y_{\tt N} =  \psi^\top \theta,
$
with 
\begequ
\label{psi_y:prop4}
\begin{aligned}
y_{\tt N} & ~:= ~ - \phi T_1\xi - G_2\big[ (\phi y^\top + \alpha \Pi_y) T_2 \xi \big]
\\
\psi &  ~:=~ \left(\phi T_1\Psi\right)^\top + G_2\big[ (\phi y^\top + \alpha \Pi_y) T_2\Psi \big]^\top ,
\end{aligned}
\endequ
$\phi$ defined in \eqref{phi:r}, the filters $G_1 [\cdot],~G_2[\cdot]$ starting from zero, the parameter $\alpha>0$, and the matrices 
$
T_1:= \begmat{1 &~ 0_{1\times 9}},~
T_2:= \begmat{0_{3\times1} ~& I_3 ~& 0_{3\times 6}}.
$
\end{proposition}
\begin{proof}\rm
From the analysis in Section \ref{sec3}, the full dynamics is given by \eqref{dot_xR}-\eqref{dot_v} and $\dot r  = - y^\top v$.
%
Define the error between $Q$ and $R$ in $SO(3)$ as
$
E(R,Q):= RQ^\top,
$
and we have
$$
\begin{aligned}
\dot{\ob{E(R,Q)}} & := \dot R Q^\top - R Q^{-1} \dot Q Q^{-1}
\\
& = R \Omega_\times Q^\top - R Q^\top Q \Omega_\times Q^\top
= 0.
\end{aligned}
$$
Hence, there exists a constant matrix $Q_c \in SO(3)$ satisfying
$$
R(t) = Q_c Q(t),~ \forall t\ge 0,
$$
with $Q_c := R(0)Q(0)^\top$. Then, we have the parameterization to the last term in the dynamics of $v$ as
$$
R(t)^\top g = Q(t)^\top g_c, ~ \forall t\ge 0
$$
with a new \emph{constant unknown} vector $g_c \in \rea^3$ defined as $g_c: =Q_c^\top g$.

Now considering that $b_a$ is a constant bias, and defining the extended state
$
\chi:= \col(r, v, b_a, g_c) \in \rea^{10},
$
the plant dynamics may be compactly rewritten as
\begin{equation}
\label{dot:chi1}
\dot \chi = A(y,\Omega,Q)\chi + B(a).
\end{equation}
Comparing the dynamics \eqref{dot:chi1} and \eqref{dyn_ext:2} and following the idea of generalized (G)PEBO \cite{ORTgpebo}, it yields
$
 \frac{d}{dt}(\chi - \xi) =A(y,\Omega, Q)\big[\chi - \xi \big].
$ 
The matrix $\Psi(t)\Psi(s)^\top$ can be viewed as the state transition matrix for this LTV system from $s$ to $t$. Hence, this yields
\begequ
\label{lre:phi}
\begin{aligned}
\chi - \xi = \Psi[ \chi_0 - \xi(0)] \quad 
\implies & \quad
\chi = \xi + \Psi \chi_0
\\
\overset{\theta:=\chi_0}{\implies} & \quad
\chi = \xi + \Psi \theta,
\end{aligned}
\endequ
where in the second equation we have used  $\xi(0)=0$.
From the proof of Proposition \ref{prop:1}, we have
\begequ
\label{equality1}
r\phi + G_2 [(\phi y^\top  + \alpha \Pi_y) v] =0
\endequ
with $\phi$ defined in \eqref{phi:r}, for the case with zero filtering initial conditions. From the definitions of $\chi$ and $X$ we have
$
X = T(y) \chi,
$
thus verifying the algebraic equation \eqref{algebraic:4.1}. Substituting the last equation of \eqref{lre:phi} into \eqref{equality1}, we then have
$$
\begin{aligned}
 &\quad	\phi T_1 \chi +G_2[(\phi y^\top + \alpha \Pi_y) T_2 \chi] =0
 \\
 \implies & \quad
 \phi T_1 (\xi+\Psi\theta) + G_2[(\phi y^\top + \alpha \Pi_y) T_2 (\xi+\Psi\theta)] =0
 \\
 \implies & \quad
  y_{\tt N} =\psi^\top \theta
\end{aligned}
$$
with the signals $y_{\tt N}$ and $\psi$ given in \eqref{psi_y:prop4}.
\qed
\end{proof}

The above regression model is motivated by an extension of PEBO from Euclidean space to matrix Lie groups in our previous work \cite{YIetalCDC}. By gathering the unknown but constant matrices $Q_c \in SO(3)$ and $g \in \rea^3$ in a new vector $g_c$, it provides an effective way to deal with the scenario with unknown attitude and gravitational constant, {\em e.g.}, in aerospace applications.


\subsection{Position-Velocity Observer Design}
\label{sec:4b}

In the above subsection we present a novel linear regression model with respect to the constant vector $\theta$, the estimation of which is sufficient to reconstruct the state $X$. Now, the remaining task is to estimate $\theta$ stemmed from Proposition \ref{prop:lre1}. In the following we provide a globally convergent position-velocity observer with biased inertial measurements for the robot trajectory \emph{not} satisfying the PE condition.

\begin{proposition}
\label{prop:5}\rm
Consider the dynamics \eqref{dot_z}-\eqref{dot_v} under Assumption \ref{ass:1}, and the filtered signal
\begequ
\label{filter_H0}
\hat{\theta} = \calh[(y_{\tt N}, \psi)]
\endequ
with $y_{\tt N},\psi$ defined in Proposition \ref{prop:lre1} and filter $\calh$ given by
\begin{equation}
\label{filter_H}
\left\{
\begin{aligned}
 \dot \Phi & = - \rho \Phi + \psi \psi^\top
\\
 \dot Y & = - \rho Y + \psi y_{\tt N}
 \\
 \dot \zeta & = \Delta \bby - \Delta^2 \zeta
\\
 \dot \omega & = - \Delta^2 \omega , \quad \omega(0)=1
 \\
 \dot{\hat \theta} & = \gamma \big[
 (\zeta + k_p\Delta \bby) - (1-\omega + k_p \Delta^2) \hat \theta
 \big]
\end{aligned}
\right.
\end{equation}
with the gains $\rho >0$, $\gamma>0$ and $k_p>0$, $\Delta:= \det\{\Phi\}$, $\bby:= \adj\{\Phi\}Y$, $n= \dim\{\theta\}$, and the initial conditions $\zeta(0)= 0_n, Y(0)=0_n$ and $\Phi(0)= 0_{n\times n}$. Then, the observer consisting of \eqref{dyn_ext:2}, \eqref{filter_H} and the estimate
$
\hat X= T(y)(\xi + \Psi \hat \theta)
$
with $T(y)$ defined in Proposition \ref{prop:lre1}, guarantees global exponential convergence
$
\lim_{t\to\infty} |\hat X(t) - X(t)| =0
$
from any initial guess $\hat\theta(0) \in \rea^n$, assuming that $\psi$ is IE.
\end{proposition}
\begin{proof}\rm
According to Proposition \ref{prop:lre1}, the system dynamics \eqref{dot_z} and \eqref{dot_v}, together with \eqref{dyn_ext:2} admits the LRE $y_{\tt N} = \psi^\top \theta$. Following the dynamic regressor extension and mixing (DREM) technique \cite{ARAetal}, after going through an LTV filter --- the first two equations in \eqref{filter_H} --- we obtain ${d\over dt}(Y-\Phi\theta) = -\rho (Y-\Phi\theta)$ with $\rho>0$ from $Y(0)-\Phi(0)\theta  =0$, thus obtaining the Kreisselmeier’s  extended LRE \cite{KRE}
\begin{equation}
\label{elre}
Y = \Phi \theta
\end{equation}
with $Y\in \rea^n$ and $\Phi \in \rea^{n\times n}$. After pre-multiplying the adjugate matrix $\adj\{\Phi\}$ to the both sides of \eqref{elre}, we get the decoupled regressors
$
\mathbb{Y} = \Delta \theta
$
with a \emph{scalar} regression $\Delta$, which is now in the same form as in Proposition \ref{prop:3}. Without loss of generality, we assume $\psi$ is IE in the interval $[0,t_c]$ with $t_c>0$, {\em i.e.},
$
\int_0^{t_c}\psi(s)\psi(s)^\top ds \succeq \delta I
$
for some $\delta >0$. From 
$
\dot \Phi = -\rho\Phi + \psi \psi^\top
$
with $\Phi(0)=0_{n\times n}$, we have
$$
\begin{aligned}
 \Phi(t_c)
 ~=~  &\int_{0}^{t_c} e^{-\rho(t_c-s)} \psi(s)\psi(s)^\top ds
\\
 ~\succeq~ & e^{-\rho t_c} \int_{0}^{t_c}  \psi(s)\psi(s)^\top ds
\\
 ~\succeq~ & \delta e^{-\rho t_c}  I.
\end{aligned}
$$
From the definition $\Delta$ being the determinant of $\Phi$, it yields 
$$
\begin{aligned}
\psi \in \mbox{ IE} \quad & \implies \quad \Delta \in \mbox{IE}
\\
 & \implies \quad (1-\omega + k_p\Delta^2) \in \mbox{PE},
\end{aligned}
$$
where the second implication can be proved using the same arguments as Proposition \ref{prop:3} thus omitted. On the other hand, we can verify that the vector $\theta$ satisfies the LRE
\begequ
\label{lre:4}
\zeta + k_p\Delta \mathbb{Y} = (1-\omega + k_p\Delta^2) \theta.
\endequ
Now, we obtain a new LRE \eqref{lre:4}, satisfying the PE condition, from the IE regressor \eqref{algebraic:4.1}. By defining the estimation error $\tilde \theta = \hat \theta - \theta$, we have
\begequ
\label{error_dyn:4}
\dot{\tilde \theta} = - \gamma (1-\omega + k_p\Delta^2) \tilde\theta.
\endequ
Since $(1-\omega + k_p\Delta^2)$ is PE and non-negative for $k_p>0$, the LTV system \eqref{error_dyn:4} is globally exponentially stable at the origin. By invoking the identity \eqref{algebraic:4.1} and state boundedness from Assumption \ref{ass:1}, it completes the proof.
\qed
\end{proof}


\section{Application to Visual-Inertial Navigation}
\label{sec5}

We now apply the proposed observer to the vision-aided inertial navigation systems as an illustration. In this problem, there are several feature points in the camera view field, whose inertial positions ${}^I z_i$ are known in advance. We assume that some data association algorithm has been used to determine if these observations correspond to landmarks in the given
map. Note that a known map makes it significantly simpler than the SLAM problem, which is of more practical interests; see for example \cite{LOUetal}. The camera may provide the bearing measurement of features in the body-fixed frame
\begequ
\label{yi}
y_i = {z_i \over |z_i|} = R^\top{ {}^I z_i - x \over |{}^I z_i - x|},
\quad
i \in \caln := \{1,\ldots, n\}.
\endequ
We adopt the same notations in Section \ref{sec2}, except the subindex for feature points to distinguish them. 

\begin{prob}\rm
 ({\em Navigation observer})
Consider \eqref{dot_xR}-\eqref{dot_v} with measurements $(a,\Omega)$ and $y_i$. Both the bias $b_a \in \rea^3$ and the gravity vector $g \in \rea^3$ are unknown.
Assume the positions ${}^I z_i$ in $\{\cali\}$ are constant and known. Design an observer to asymptotically estimate the pose $(R,x)$.
\end{prob}

A simple solution is to use the proposed range observer to reconstruct the position of feature points in $\{\calb\}$, and then solve the localization problem with ``full position measurement''. We need the following.

\begin{assumption}
\label{ass:three}
There exist $i,j\in \caln \backslash\{ n\}$ such that
$
{}^I \eta_i \times {}^I \eta_j \neq 0, ~ i\neq j
$
with the vectors ${}^I \eta_i:= {}^I z_{i+1} - {}^I z_{i}$.
\end{assumption}

\begin{proposition}
\rm\label{prop:navigation_observer}
Consider the navigation observer consisting of the ranges observer
\begin{align}\nonumber
\dot Q & ~=~ Q\Omega_\times
\\
\dot \xi & ~=~ A_e(\Omega,Q, \bar y )\xi + B_e(a), \quad \xi(0)=0
\\ \label{obs:navigation1}
\dot \Psi &~ =~ A_e(\Omega,Q, \bar y)\Psi , \quad \Psi(0)=I
\\ \nonumber
\hat\theta & ~= ~\calh(y_{\tt N},\psi), ~
\hat v  = T_v(\xi + \Psi \hat\theta),
\\
\hat z_i & ~=~ T_{z,i}(\xi + \Psi \hat\theta)  y_i
\end{align}
with $\bar y:= \col(y_1^\top, \ldots, y_n^\top)$ and
\begin{equation}
\label{psi_y:navigation}
\begin{aligned}
	y_{\tt N}& ~:=~ - \Lambda_\phi T_r\xi - G_2 [(\Lambda_\phi\bar y+ \alpha \mathbf{\Pi}) T_v\xi ]
	\\
	\psi &~:=~(\Lambda_\phi T_r\Psi)^\top + G_2[ (\Lambda_\phi\bar y + \alpha \mathbf{\Pi})T_v \Psi ]^\top,
\end{aligned}
\end{equation}
the filter $\calh$ defined in \eqref{filter_H0}-\eqref{filter_H}, $\mathbf{\Pi}:= \col(\Pi_{y_1}, \ldots, \Pi_{y_n})$, the matrices $T_r:= [0_{n\times 9}~ I_n]$, $T_v := [I_3 \quad 0_{3\times (6+n)}]$, $T_{z,i}:= [0_{1\times 9} \quad  \mathsf{e}_i^\top]$, $\mathsf{e}_i$ being the $i$-th canonical basis in $\rea^n$, $\Lambda_\phi := \diag\{\phi_1,\ldots, \phi_n\}$, $\phi_i$ in \eqref{phi:r} for the $i$-th feature point, and
%
\begin{equation}
\label{AeBe}
A_e:=
\begmat{-\Omega_\times & I_3 & Q^\top &  0_{3\times n}
\\
\mbox{---}&&0_{6\times (n+9)} & \mbox{---}
\\
- y_1^\top & 0 & &
\\
\vdots & &\ddots&
\\
- y_n^\top & & & 0
},
~
B_e := \begmat{a \\ 0 \\ \vdots \\ 0 },
\end{equation}
cascading to the (full-position) localization observer
\begin{align} \nonumber
\dot{\hat Q}_c & ~ =~ - w_\times \hat Q_c , \quad  \hat R  = \hat Q_c Q,
\\ \label{obs:navigation2} 
w & ~= ~ {1\over 2}{\textstyle\sum_{i \in \caln \backslash\{n\}}} k_i {}^I \eta_i \times (\hat Q_c Q \hat \eta_i)
\\
\dot{\hat x} & ~=~ \hat Q_c Q \hat v +  {\textstyle \sum_{i \in \caln}\sigma_i} ({}^I z_i - \hat x - \hat Q_c Q \hat z_i )\nonumber
\end{align}
with $\hat \eta_i = \hat z_{i+1} - \hat z_i$ and the gains $k_i,\sigma_i, \alpha>0$. If $\psi$ is IE and Assumption \ref{ass:three} holds, then the convergence
$
\lim_{t\to \infty} [\|\hat R(t) - R(t)\| + |\hat x(t) - x(t)|] =0
$
holds almost globally.
\end{proposition}
\begin{proof}\rm It is given in Appendix.\qed
\end{proof}

The proposed navigation observer is implemented in a \emph{modular} manner, {\em i.e.}, it constitutes of a ranges observer \eqref{obs:navigation1} and a full-position localization observer \eqref{obs:navigation2}. To be precise, we first use \eqref{obs:navigation1} to reconstruct the full position ``measurement'', which is then utilized in the localization observer \eqref{obs:navigation2}. We figure out that the ``integrated'' ranges observer \eqref{obs:navigation1} can be replaced by several ``individual'' observers in Propositions \ref{prop:lre1}-\ref{prop:5} for each feature point. 


\section{Simulation Results}
\label{sec6}

%
\begin{figure*}[!htp]
   \centering
   \subfigure[Range estimate $\hat r$ (IE)]{
   \includegraphics[width=0.27\textwidth]{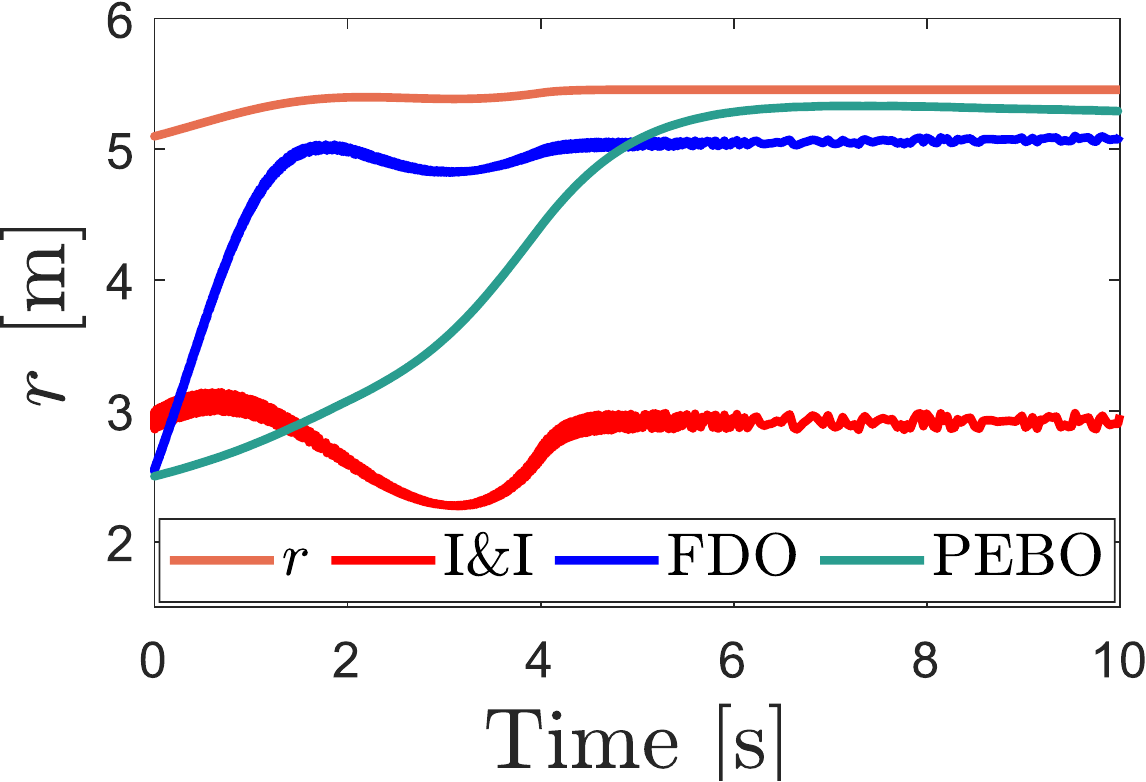}
   \label{fig:with-v1-10}
   }
   \subfigure[Position estimate error $|\hat z - z|$ (PE)]{
   \includegraphics[width=0.27\textwidth]{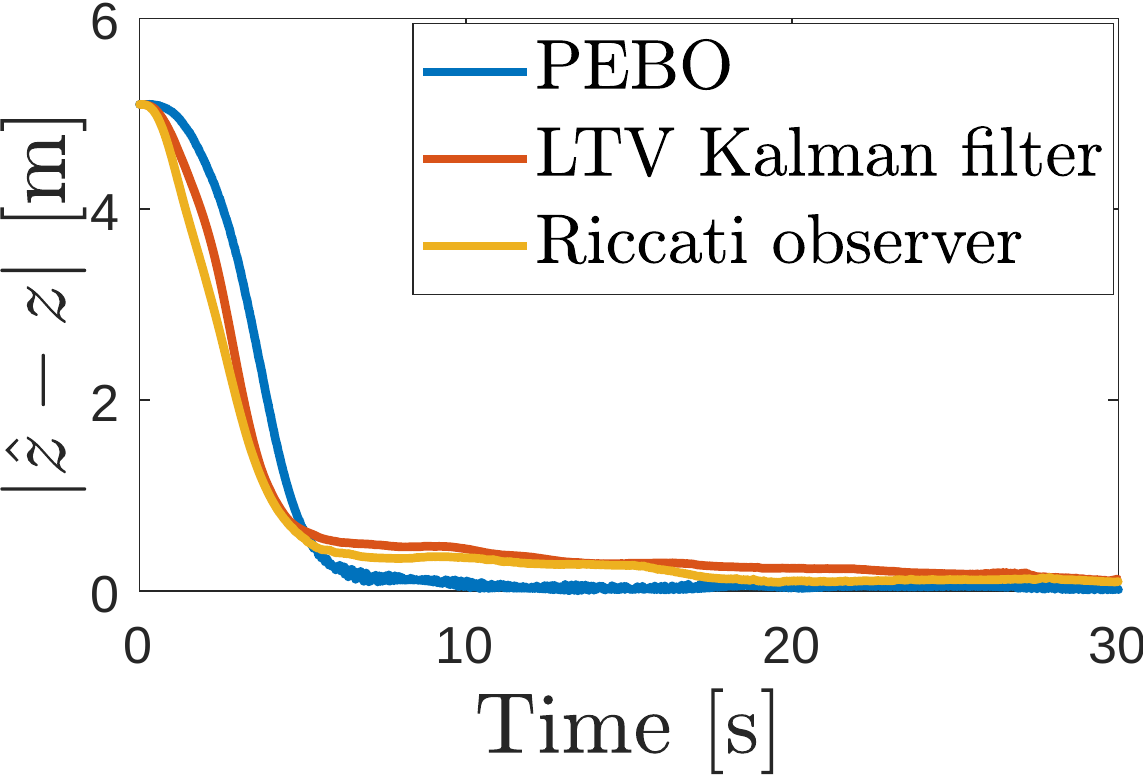}
   \label{fig:riccati1}
   }
   \subfigure[Position estimate error $|\hat z -z| $ (IE)]{
   \includegraphics[width=0.27\textwidth]{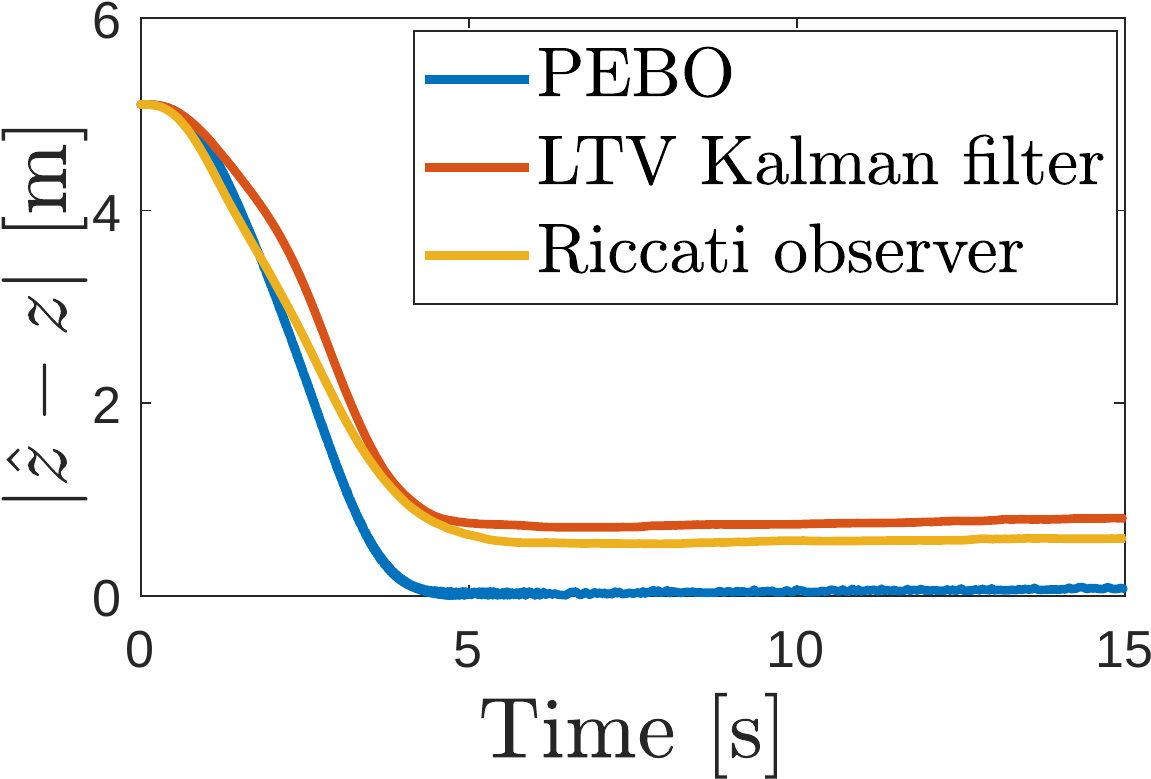}
   \label{fig:riccati2}
   }
   \caption{Comparison among the PEBO, the I\&I observer in \cite{ASTetal}, the feature depth observer (FDO) in \cite{DELetal}, the LTV Kalman filter in \cite{LOUetal} and the Riccati observer in \cite{HAMSAM} with linear velocity information (with noise)}
    \label{fig:comparison}
\end{figure*}
\begin{figure*}[!htp]
   \centering
   \subfigure[Range estimate $\hat r$]{
   \includegraphics[width=0.27\textwidth]{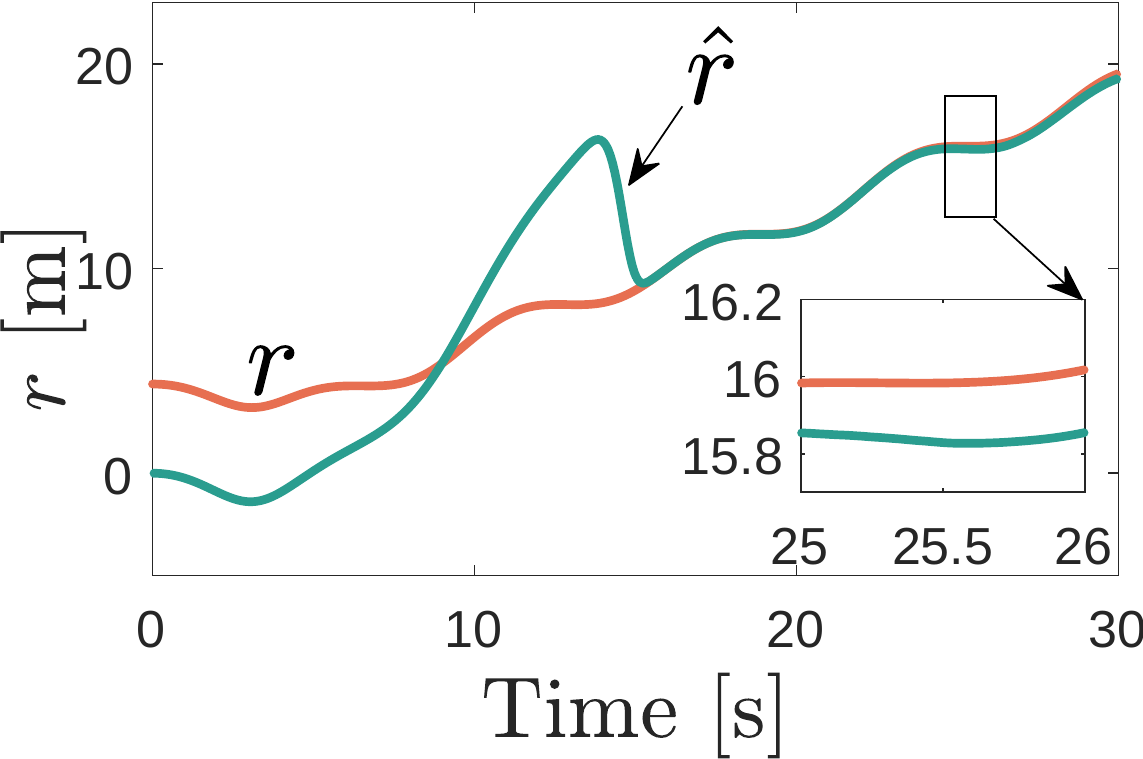}
   \label{fig:a5}
   }
   \subfigure[{Velocity estimation error}]{
   \includegraphics[width=0.27\textwidth]{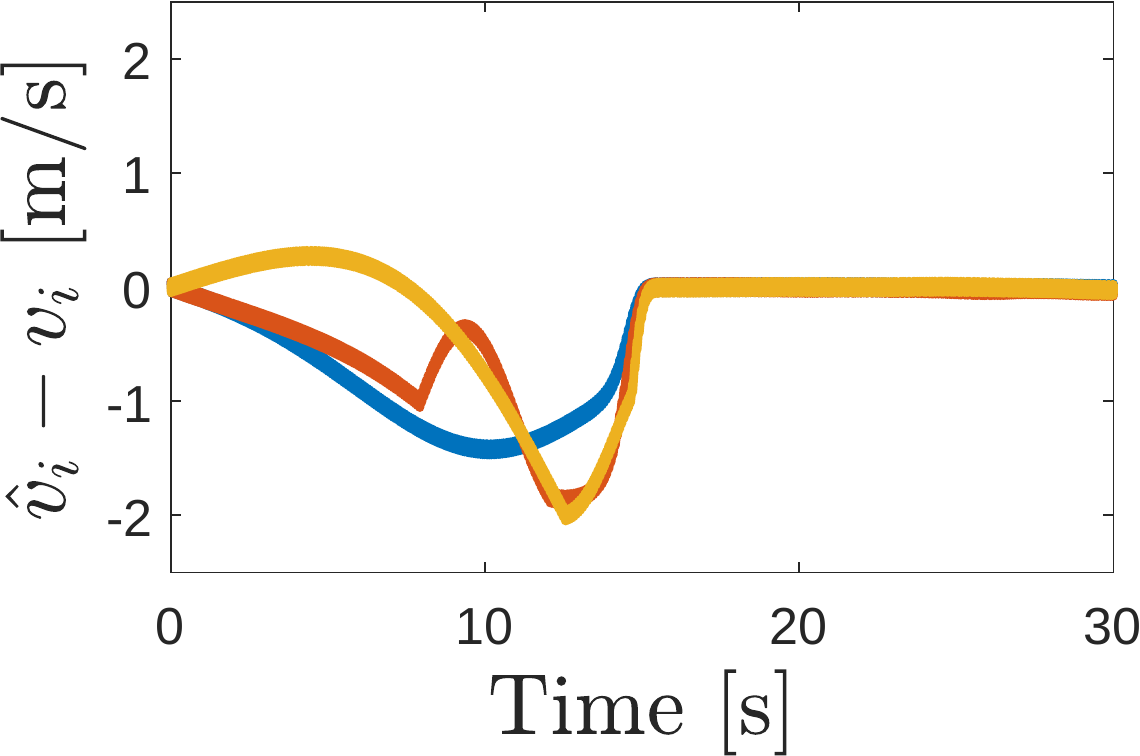}
   \label{fig:a7}
   }
   \subfigure[Bias estimation error]{
   \includegraphics[width=0.27\textwidth]{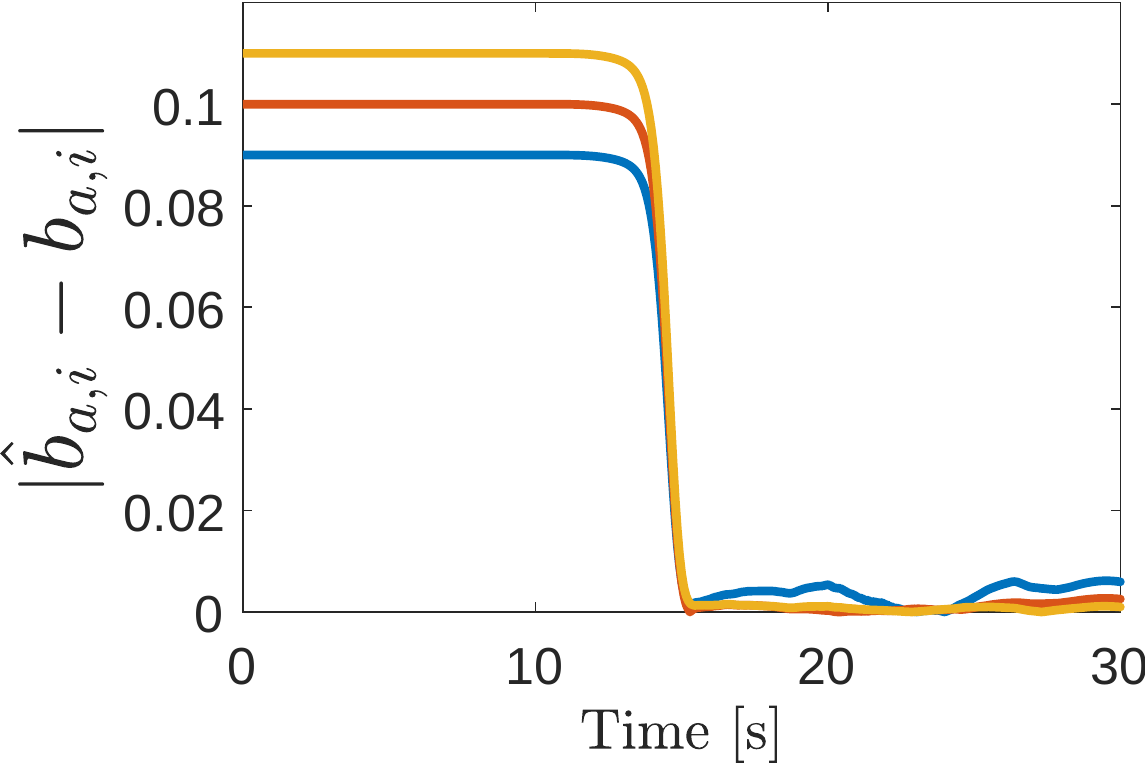}
   \label{fig:a8}
   }
\caption{Performance of the proposed position-velocity observer with biased acceleration measurement (IE, with noise)}
    \label{fig:6}
\end{figure*}
\begin{figure*}[!htp]\small
   \centering
   \subfigure[Position estimation error]{
   \includegraphics[width=0.27\textwidth]{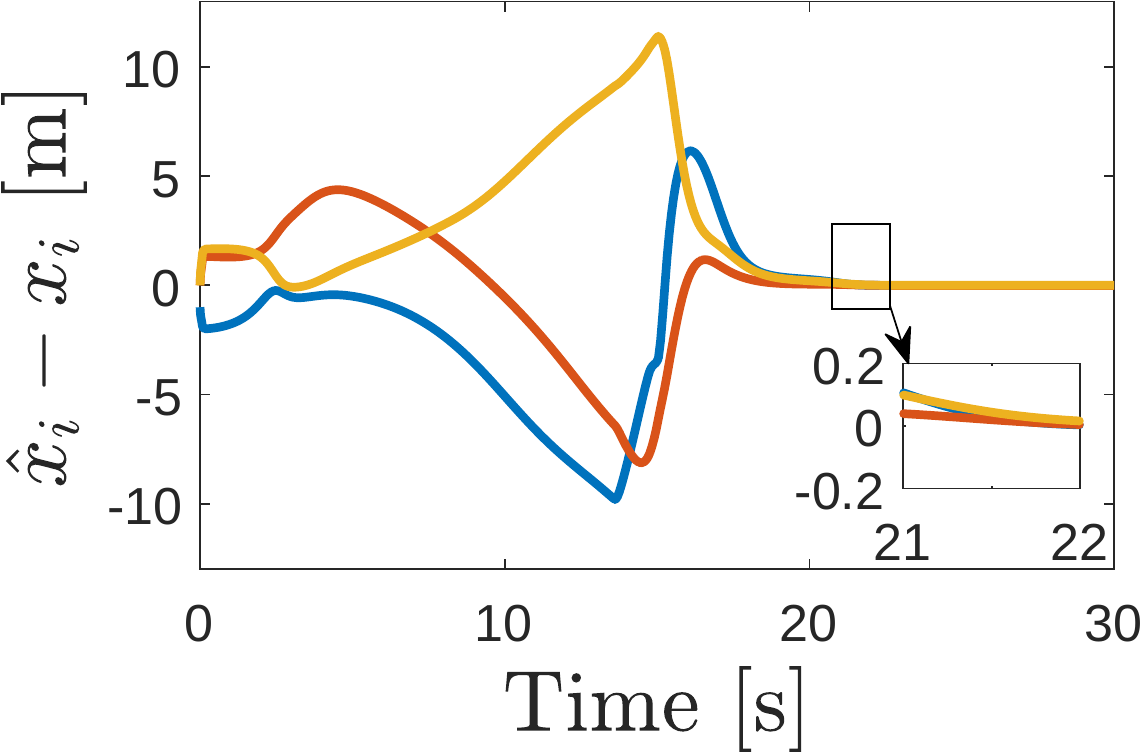}
   \label{fig:n3}
   }
   %
   %
   \subfigure[Attitude estimation error]{
   \includegraphics[width=0.27\textwidth]{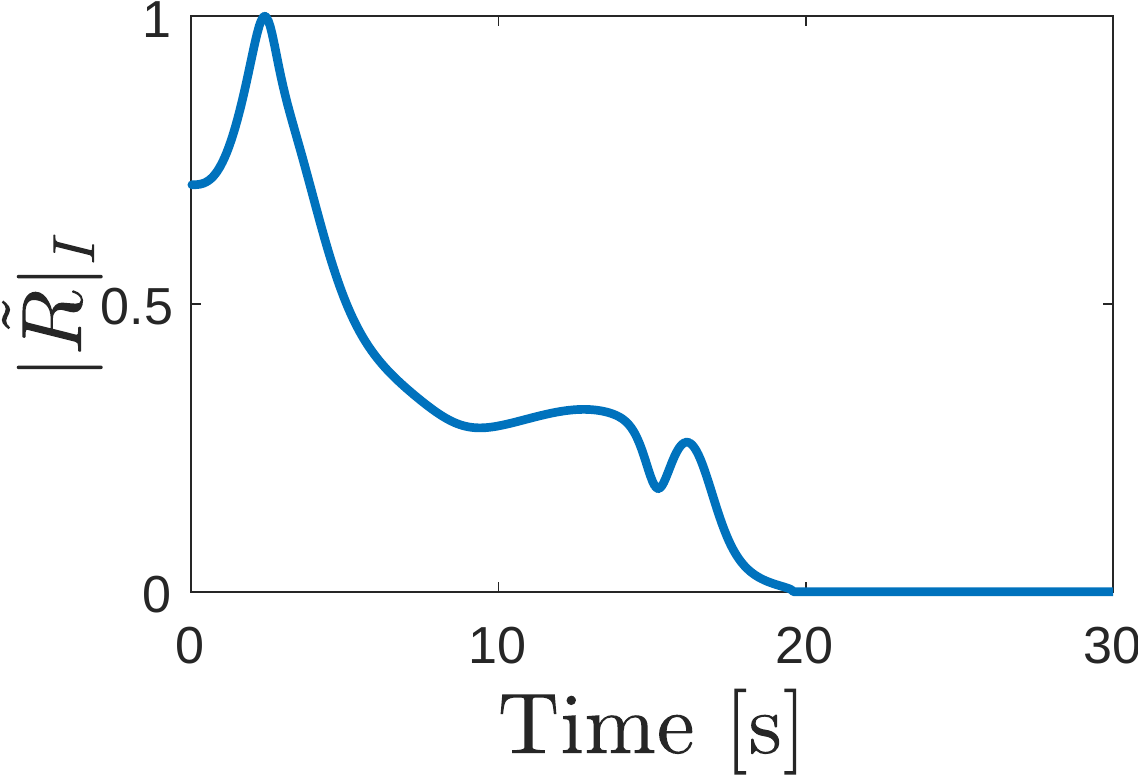}
   \label{fig:n1}
   }
  \subfigure[Velocity estimation error]{
   \includegraphics[width=0.27\textwidth]{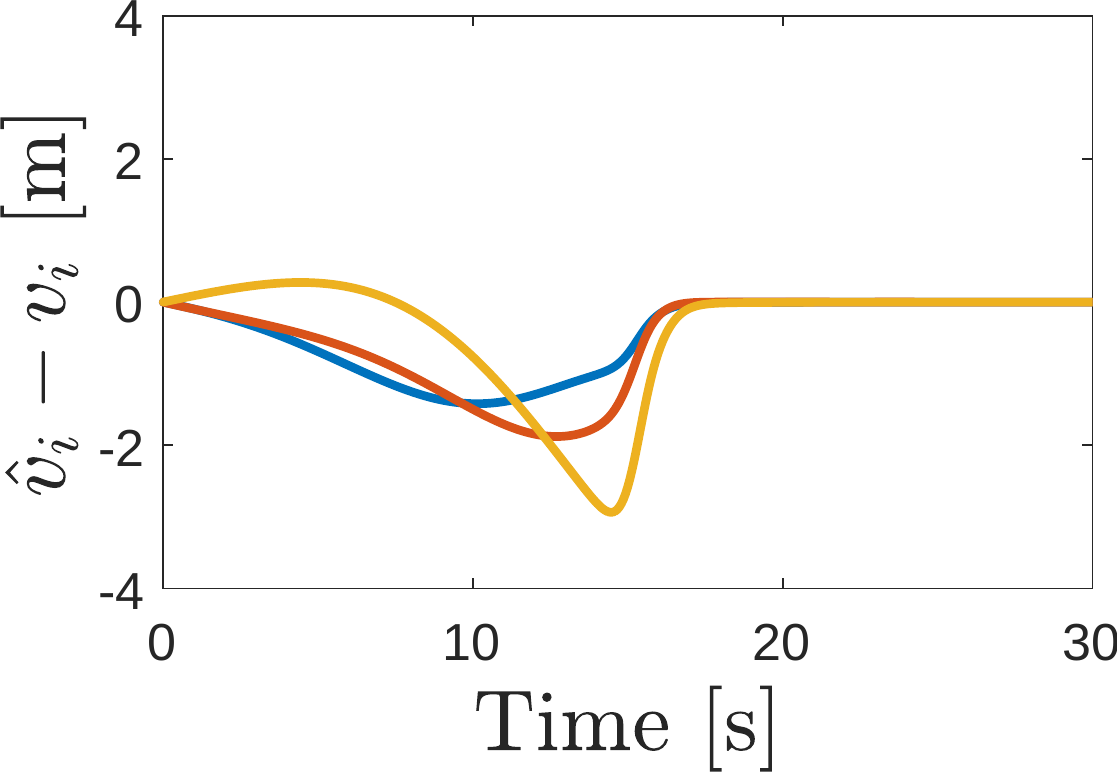}
   \label{fig:n2}
   }
\caption{Performance of the proposed navigation observer with biased acceleration measurement (IE, noise-free)}
    \label{fig:7}
\end{figure*}

In this section, we present simulation results to illustrate the properties of the proposed observers. Full details and additional figures can be found in Appendix. First, consider the case of measurable velocity in Section \ref{sec3} with two trajectories, one being PE and the other being IE but not PE. In Fig. \ref{fig:comparison}, we observe that the estimate from the PEBO with measurement noise ultimately converges to a small neighbourhood of its true value, showing robustness \emph{vis-\`a-vis} noise. We also simulated the I\&I observer \cite{KARAST}, the feature depth observer \cite{DELetal}, the LTV Kalman filter \cite{LOUetal} and the Riccati observer \cite{HAMSAM} all of which require some PE conditions. Since the observers in \cite{DELetal,KARAST} are concerned with the depth estimation, we draw them together for comparison using the IE trajectory in Fig. \ref{fig:with-v1-10}. These observers have different state spaces, so we chose the initial range estimates as close as possible to make a fair comparison. Note that both the I\&I and the feature depth observers have large ultimate errors since their excitation conditions are not satisfied. In Figs. \ref{fig:riccati1}-\ref{fig:riccati2}, we compare to the position observers in \cite{LOUetal} and \cite{HAMSAM} using both the PE and IE trajectories. As expected, these three observers admit similar accuracy for PE trajectories; however, for IE trajectories, the proposed PEBO is superior to the others.

Now let us consider position and velocity estimation in Section \ref{sec4}, testing with an IE but not PE trajectory and noisy measurements. The initial conditions for the observer are set as $\hat \theta(0)= [0_{1\times 9}, 10]$ with a small sensor bias. Fig. \ref{fig:6} shows convergence of all estimation errors to a neighborhood of zero.

Lastly, we evaluate the performance of the navigation observer in Proposition \ref{prop:navigation_observer} with three known landmarks in $\{\cali\}$. The same robot trajectory and conditions were adopted. The behaviour is shown in Fig. \ref{fig:7} in the absence of measurement noise. The pose (including the position and the attitude) estimate converge its true value after a short transient stage, showing high estimation performance.

%
\section{Concluding Remarks}
\label{sec7}
We have proposed a nonlinear observer for the problem of range estimation of a visual feature, using only its bearing and linear acceleration measurements, by means of the recently developed PEBO methodology. The design is applicable to the case with unknown sensor bias and gravitational constant, and can achieve global exponential convergence under a weak IE condition on the robot trajectory. The trade off is that the design requires a more complicated observer structure than other methods. To illustrate its utility, we apply the observer to the problem of visual inertial navigation, providing a novel almost globally convergent solution. 
\bibliographystyle{abbrv}        
\bibliography{reference}

\begin{thebibliography}{10}

\bibitem{ARAetal}
S.~Aranovskiy, A.~Bobtsov, R.~Ortega, and A.~Pyrkin.
\newblock Performance enhancement of parameter estimators via dynamic regressor
  extension and mixing.
\newblock {\em IEEE Trans. Autom. Control}, 62(7):3546--3550, 2016.

\bibitem{ASTetal02}
A.~Astolfi, L.~Hsu, M.~S. Netto, and R.~Ortega.
\newblock Two solutions to the adaptive visual servoing problem.
\newblock {\em IEEE Trans. Rob. Autom.}, 18(3):387--392, 2002.

\bibitem{ASTetal}
A.~Astolfi, D.~Karagiannis, and R.~Ortega.
\newblock {\em Nonlinear and Adaptive Control with Applications}.
\newblock Springer, 2007.

\bibitem{BERetal}
S.~Berkane, A.~Tayebi, and S.~De~Marco.
\newblock A nonlinear navigation observer using {IMU} and generic position
  information.
\newblock {\em Automatica}, 127:109513, 2021.

\bibitem{BES}
G.~Besan{\c{c}}on.
\newblock {\em Nonlinear Observers and Applications}.
\newblock Springer, 2007.

\bibitem{BOBetal}
A.~Bobtsov, B.~Yi, R.~Ortega, and A.~Astolfi.
\newblock Generation of new exciting regressors for consistent on-line
  estimation of unknown constant parameters.
\newblock {\em IEEE Trans. Autom. Control}, 2022.

\bibitem{DANetal}
A.~P. Dani, N.~R. Fischer, and W.~E. Dixon.
\newblock Single camera structure and motion.
\newblock {\em IEEE Trans. Autom. Control}, 57(1):238--243, 2011.

\bibitem{DELetal}
A.~De~Luca, G.~Oriolo, and P.~Robuffo~Giordano.
\newblock Feature depth observation for image-based visual servoing: Theory and
  experiments.
\newblock {\em Int. J. Rob. Res.}, 27(10):1093--1116, 2008.

\bibitem{DIXetal}
W.~E. Dixon, Y.~Fang, D.~M. Dawson, and T.~J. Flynn.
\newblock Range identification for perspective vision systems.
\newblock {\em IEEE Trans. Autom. Control}, 48(12):2232--2238, 2003.

\bibitem{HAMSAM}
T.~Hamel and C.~Samson.
\newblock Position estimation from direction or range measurements.
\newblock {\em Automatica}, 82:137--144, 2017.

\bibitem{HANKAN}
M.~Han and T.~Kanade.
\newblock Reconstruction of a scene with multiple linearly moving objects.
\newblock {\em Int. J. Comput. Vis.}, 59(3):285--300, 2004.

\bibitem{HUetal}
G.~Hu, D.~Aiken, S.~Gupta, and W.~E. Dixon.
\newblock Lyapunov-based range identification for paracatadioptric systems.
\newblock {\em IEEE Trans. Autom. Control}, 53(7):1775--1781, 2008.

\bibitem{JANGHO}
M.~Jankovic and B.~K. Ghosh.
\newblock Visually guided ranging from observations of points, lines and curves
  via an identifier based nonlinear observer.
\newblock {\em Syst. Control Lett.}, 25(1):63--73, 1995.

\bibitem{KARAST}
D.~Karagiannis and A.~Astolfi.
\newblock A new solution to the problem of range identification in perspective
  vision systems.
\newblock {\em IEEE Trans. Autom. Control}, 50(12):2074--2077, 2005.

\bibitem{KRE}
G.~Kreisselmeier.
\newblock Adaptive observers with exponential rate of convergence.
\newblock {\em IEEE Trans. Autom. Control}, 22(1):2--8, 1977.

\bibitem{LAGetal}
C.~Lageman, J.~Trumpf, and R.~Mahony.
\newblock Gradient-like observers for invariant dynamics on a {Lie} group.
\newblock {\em IEEE Trans. Autom. Control}, 55(2):367--377, 2009.

\bibitem{LOUetal}
P.~Louren{\c{c}}o, P.~Batista, P.~Oliveira, and C.~Silvestre.
\newblock A globally exponentially stable filter for bearing-only simultaneous
  localization and mapping with monocular vision.
\newblock {\em Robot. Auton. Syst.}, 100:61--77, 2018.

\bibitem{MAHetal}
R.~Mahony, T.~Hamel, and J.-M. Pflimlin.
\newblock Nonlinear complementary filters on the special orthogonal group.
\newblock {\em IEEE Trans. Autom. Control}, 53(5):1203--1218, 2008.

\bibitem{MOUROU}
A.~I. Mourikis and S.~I. Roumeliotis.
\newblock A multi-state constraint kalman filter for vision-aided inertial
  navigation.
\newblock In {\em IEEE Int. Conf. Rob. Autom.}, pages 3565--3572. IEEE, 2007.

\bibitem{ORTgpebo}
R.~Ortega, A.~Bobtsov, N.~Nikolaev, J.~Schiffer, and D.~Dochain.
\newblock Generalized parameter estimation-based observers: Application to
  power systems and chemical--biological reactors.
\newblock {\em Automatica}, 129:109635, 2021.

\bibitem{ORTetalscl}
R.~Ortega, A.~Bobtsov, A.~Pyrkin, and S.~Aranovskiy.
\newblock A parameter estimation approach to state observation of nonlinear
  systems.
\newblock {\em Syst. Control Lett.}, 85:84--94, 2015.

\bibitem{ORTYI}
R.~Ortega, B.~Yi, S.~Vukosavi{\'c}, K.~Nam, and J.~Choi.
\newblock A globally exponentially stable position observer for interior
  permanent magnet synchronous motors.
\newblock {\em Automatica}, 125:109424, 2021.

\bibitem{SASetal}
M.~Sassano, D.~Carnevale, and A.~Astolfi.
\newblock Observer design for range and orientation identification.
\newblock {\em Automatica}, 46(8):1369--1375, 2010.

\bibitem{SASBOD}
S.~Sastry and M.~Bodson.
\newblock {\em Adaptive Control: Stability, Convergence And Robustness}.
\newblock Courier Corporation, 2011.

\bibitem{SHI}
K.~Shimizu.
\newblock Nonlinear state observers by gradient descent method.
\newblock In {\em IEEE Int. Conf. Control Appl.}, pages 616--622, 2000.

\bibitem{SPIetal}
R.~Spica, P.~R. Giordano, and F.~Chaumette.
\newblock Active structure from motion: Application to point, sphere, and
  cylinder.
\newblock {\em IEEE Trans. Robot.}, 30(6):1499--1513, 2014.

\bibitem{TAHetal}
O.~Tahri, D.~Boutat, and Y.~Mezouar.
\newblock Brunovsky's linear form of incremental structure from motion.
\newblock {\em IEEE Trans. Robot.}, 33(6):1491--1499, 2017.

\bibitem{TAO}
G.~Tao.
\newblock {\em Adaptive Control Design and Analysis}.
\newblock John Wiley \& Sons, 2003.

\bibitem{WANORTetal}
L.~Wang, R.~Ortega, A.~Bobtsov, J.~G. Romero, and B.~Yi.
\newblock Identifiability implies robust, globally exponentially convergent
  on-line parameter estimation: Application to model reference adaptive
  control.
\newblock {\em ArXiv Preprint}, 2021.

\bibitem{WANetal}
M.~Wang, S.~Berkane, and A.~Tayebi.
\newblock Nonlinear observers design for vision-aided inertial navigation
  systems.
\newblock {\em IEEE Trans. Autom. Control}, 2021.

\bibitem{YIetalCDC}
B.~Yi, C.~Jin, L.~Wang, G.~Shi, and I.~R. Manchester.
\newblock An almost globally convergent observer for visual {SLAM} without
  persistent excitation.
\newblock In {\em {IEEE} {Conf.} {Decis.} {Control}}, pages 5441--5446, Dec.
  2021.

\bibitem{ZHAZEL}
S.~Zhao and D.~Zelazo.
\newblock Bearing rigidity and almost global bearing-only formation
  stabilization.
\newblock {\em IEEE Trans. Autom. Control}, 61(5):1255--1268, 2015.

\end{thebibliography}


\begin{thebibliography}{aaaaa}

\bibitem[Kha02]{KHA1}
H. K. Khalil. {\em Nonlinear Systems}, Third edition, 2002.

\bibitem[AW00]{AULWAN2}
B. Aulbach and T. Wanner. The Hartman-Grobman theorem for
carath\'eeodory-type differential equations in Banach spaces. {\em Nonlinear
Anal. Theory Methods Appl.}, 40(1-8):91–104, 2000.

\end{thebibliography}

\clearpage
\newpage
\appendix
\section*{Appendix}


\section{Proof of Proposition \ref{prop:2}}

{\em Proof.} Define the estimation error of range as $\tilde r := \hat r - r$ with $r=|z|$, the dynamics of which is given by
$$
\dot{\tilde r} = - \gamma\phi^\top \left(\phi \hat r + G_2[y^\top v \phi] + \alpha G_2[\Pi_y v]\right).
$$
Invoking \eqref{filter_r}, we have
\begequ
\label{dot:tilde_r}
\dot {\tilde r} =- \gamma \phi^\top \phi \tilde r + \et,
\endequ
in which the exponentially decaying term $\et$ is caused by initial conditions of the stable filters $G_1[\cdot]$ and $G_2[\cdot]$. 

Now we use the standard perturbation analysis to show that $\tilde r \to 0 $ as $t\to \infty$. Consider an auxiliary (nominal) system 
\begequ
\label{syst:nominal}
\dot {\tilde r} = - \gamma \phi^\top \phi \tilde r,
\endequ
for which the necessary and sufficient condition of the global exponential stability (GES) of the origin is that $\phi^\top$ is PE \cite[Thm 2.5.1]{SASBOD}. For the LTV system \eqref{syst:nominal}, invoking the converse Lyapunov theorem \cite[Thm 4.14]{KHA1}, there exist a smooth function $V(\tilde r,t)$ and constants $c_i \in \rea_{>0}$ ($i=1,\ldots,4$) satisfying
\begin{equation}
\label{eq:a.3}
\begin{aligned}
c_1 |\tilde r|^2 \le V(\tilde r, t) & \le c_2 |\tilde r|^2
\\
{\partial V \over \partial t} + {\partial V \over \partial \tilde r}(-\gamma \phi^\top \phi \tilde r) &  \le - c_3 |\tilde r|^2
\\
\left| {\partial V \over \partial \tilde r} (\tilde r,t) \right| & \le c_4 |\tilde r|.
\end{aligned}
\end{equation}
Consider the Lyapunov-like function 
$$
W(\tilde r,t) = V(\tilde r,t) + V_\varepsilon(t)
$$
with 
$$
V_\varepsilon:=   \int_t^{+\infty} {c_4^2 \over 2c_3} |\et(s)|^2 ds,
$$
which is well-defined since $\et$ is exponentially decaying --- thus square integrable. Calculating the time derivative of $W(\tilde r,t)$ along the trajectory of the \emph{perturbed} system \eqref{dot:tilde_r}, we have
$$
\begin{aligned}
\dot W & = {\partial V \over \partial t} + {\partial V \over \partial \tilde r}(-\gamma \phi^\top \phi \tilde r + \et) - {c_4^2 \over 2c_3}|\et(t)|^2
\\
& \le - c_3|\tilde r|^2 + \left|{\partial V \over \partial \tilde r} \right| |\et|  - {c_4^2 \over 2c_3}|\et|^2
\\
& \le 
-{c_3 \over 2}|\tilde r|^2
\\
& \le
- {c_3 \over 2c_2 } (W - V_\varepsilon(t)), 
\end{aligned}
$$
where in the third inequality we have used $|{\partial V \over \partial \tilde r} | |\et| \le {c_3\over 2} |\tilde r|^2 + {c_4^2 \over 2c_3} |\et|^2$, as well as the last equality in \eqref{eq:a.3}. Using the comparison lemma to the inequality $\dot W \le - {c_3 \over 2c_2}(W- V_\varepsilon(t) )$ and the fact $V_\varepsilon\to 0$ exponentially due to its definition, we have both $W(\tilde r ,t) $ and $\tilde r$ converging to zero exponentially fast as $t\to\infty$. Invoking the algebraic relation $z = ry$, as well as the boundedness assumption of $y$, it completes the proof.
\qed


\section{Proof of Proposition \ref{prop:navigation_observer}}
{\em Proof.} 
In the problem set of navigation, the full system dynamics is given by
\begin{equation}
\label{dyn:full_navigation}
\begin{aligned}
\dot R & = R\Omega_\times
\\
\dot v & = - \Omega_\times v + a + b_a + R^\top g
\\
\dot r_1 & = - y_1^\top v
\\
& ~\vdots
\\
\dot r_n & = - y_n^\top v
\end{aligned}
\end{equation}
With the dynamic extension $\dot Q = Q\Omega_\times$, we define the constant vector
$$
g_c:= Q(0)R(0)^\top g,
$$
and the full-state variable
$$
\chi:=\begmat{v^\top & b_a^\top & g_c^\top & r_1 & \ldots & r_n}^\top
$$
with $r_i:= |z_i|$ ($i \in \caln$). Then, we have
$$
\dot \chi = A_e(\Omega,Q, \bar y) \chi + B_e(a)
$$
with the matrices $A_e$ and $B_e$ defined in \eqref{AeBe}. It yields
$$
\dot{\ob{\chi - \xi}} ~=~ A(\bar y,\Omega,Q)[\chi - \xi],
$$
thus $\chi= \xi + \Psi\theta$ with $\theta:= \chi_0$.

From \eqref{equality1}, we have
$$
\Lambda_\phi T_r (\xi + \Psi\theta) + G_2
\left[
(\Lambda_\phi \bar y + \alpha \mathbf{\Pi}) T_v(\xi + \Psi\theta)
\right]
=0
$$
and 
$$
\bar y= 
\begmat{y_1^\top \\ \vdots \\ y_n^\top} , \quad 
\mathbf{\Pi}= \begmat{\Pi_{y_1} \\ \vdots \\ \Pi_{y_n}} 
$$ 
in which we have defined $\phi_i$ as \eqref{phi:r} for the $i$-th feature point with a slight abuse of notation using subscripts. Then, we obtain the linear regressor
$
y_{\tt N} = \psi^\top \theta
$
with \eqref{psi_y:navigation}. Following Proposition \ref{prop:5}, if $\psi$ is IE, then $\hat \theta= \calh[y_{\tt N}, \psi]$ provides a globally exponentially convergent estimate to $\theta$. Invoking state boundedness from Assumption \ref{ass:1}, it yields for all $i \in \caln$
\begequ
\label{conv:v_zi}
\lim_{t\to\infty}\left|
\begmat{\hat v(t) - v(t) \\ \hat z_i(t) - z_i(t)}
\right| =0
\quad
\mbox{(exp.)}.
\endequ

The second step is to show that $\hat Q_c \in SO(3)$ is an almost global asymptotic estimate to $Q_c$, which is similar to attitude estimation in \cite{MAHetal} but with a constant matrix $ Q_c$ rather than a time-varying one. Define its estimation error as
$$
\tilde Q_c:= Q_c \hat Q_c^\top,
$$
and its time derivative is given by\footnote{The operator $\skew(\cdot)$ is defined as $\skew(A) := \hal (A- A^\top)$ for a square matrix $A$.}
\begin{equation}
\label{dot_tilde_Q_c}
\begin{aligned}
\dot{\tilde Q}_c
& = \tilde Q_c \sum_{i=1}^{n-1} {k_i\over 2}  [(\hat Q_c Q \hat \eta_i)({}^I \eta_i)^\top - {}^I \eta_i (\hat Q_c Q \hat \eta_i)^\top]
\\
& = \tilde Q_c \sum_{i=1}^{n-1}k_i \skew [\tilde Q_c^\top  {}^I \hat \eta_i  ({}^I \eta_i)^\top]
\\\
& = \tilde Q_c \skew [\tilde Q_c^\top M] + \cale(t)
\end{aligned}
\end{equation}
where in the second equation we have used the fact $[c\times d]_\times = d c^\top - c d^\top$ for any $c,d \in \rea^n$, and define the constant matrix $M:= \sum_{i=1}^{n-1}k_i {}^I \eta_i ({}^I \eta_i)^\top$, the vectors ${}^I \hat \eta_i = R \hat \eta_i$, and
$$
\cale:= \tilde Q_c \sum_{i=1}^{n-1}k_i \skew [\tilde Q_c^\top R \tilde \eta_i ({}^I \eta_i)^\top]
$$
with exponentially decaying terms $\tilde \eta_i := \hat \eta_{i} - \eta_i$, which satisfy $\eta_i \to 0$ as $t\to 0$ exponentially fast invoking \eqref{conv:v_zi}. Since both $\tilde Q_c$ and $R$ living in a compact space $SO(3)$ and ${}^I \eta_i$ being constant, there exist two constants $a_0,a_1>0$ such that
\begin{equation}
\label{cale_ineq}
|\tr(\cale^\top(t) M)| \le a_0 e^{-a_1t} =: \varepsilon(t).
\end{equation}
We consider the time-varying Lyapunov function
$$
V(\tilde Q_c,t) = \sum_{i=1}^{n -1} k_i |{}^I \eta_i|^2 - \tr(\tilde Q_c^\top M) +  \int_t^\infty | \varepsilon (s)| ds
$$
which is a non-negative function and well defined, since $\varepsilon$ is absolutely integrable, {\em i.e.}
$$
\int_0^\infty |\varepsilon(s)|ds < +\infty.
$$
Now let us show non-negativeness of $V(\tilde{Q}_c,t)$. Since the last term $\int_t^\infty | \varepsilon (s)| ds \ge 0$, we focus on the first two terms. According to the definition of matrix $M$, invoking that ${}^I\eta_i$ are constant vectors, we have
$$
\begin{aligned}
\sum_{i=1}^{n -1} k_i |{}^I \eta_i|^2 & - \tr(\tilde Q_c^\top M)
\\ &
=
\sum_{i=1}^{n -1} k_i |{}^I \eta_i|^2 \left[
1 - \tr(
\tilde Q_c^\top   {{}^I \eta_i \over |{}^I \eta_i|} 
({{}^I \eta_i \over |{}^I \eta_i|} )^\top
)
\right].
\end{aligned}
$$
Noting that $\tilde Q_c \in SO(3)$ and ${{}^I \eta_i \over |{}^I \eta_i|}$ is a unit vector, we have $ \tr(
\tilde Q_c^\top   {{}^I \eta_i \over |{}^I \eta_i|} 
({{}^I \eta_i \over |{}^I \eta_i|} )^\top
) \le 1$. As a result, we verify that $\tilde V(\tilde Q_c,t) \ge 0$ .

The time derivative of $V(\tilde Q_c,t)$ is given by
$$
\begin{aligned}
\dot V & = - \tr\left( \skew(\tilde Q_c^\top M)^\top \tilde Q_c^\top M + \cale^\top  M \right) - |\varepsilon(t)|
\\
&= - \big\| \skew \big(\tilde Q_c^\top M \big) \big\|^2 - \tr\big(\cale^\top M\big)
 -|\varepsilon(t)|
\\
&\le - \big\| \skew \big(\tilde Q_c^\top M \big) \big\|^2,
\end{aligned}
$$
in which we have used \eqref{cale_ineq} in the last inequality. Then, it yields
$$
\int_0^\infty \big\| \skew \big(\tilde Q_c^\top(s) M \big) \big\|^2 ds < +\infty.
$$
Invoking the boundedness of the time derivative of $ \skew \big(\tilde Q_c^\top M)$ and using Babalat's lemma, we conclude that all the trajectories converge to the invariant set
$$
\Omega_e := \{\tilde Q_c \in SO(3) ~|~ \skew \big(\tilde Q_c^\top M \big) =0\}.
$$
Following the similar procedure of the proof in \cite[Theorem 5.1]{MAHetal}, we can show that the set $\Omega_e$ has a locally exponentially stable equilibrium $\tilde Q_c = I_3$, and other three isolated unstable equilibria $\bq_i$ ($i=1,2,3$) on $SO(3)$, and
there are no poles on the imaginary axis. By using the non-autonomous version of Hartman-Grobman theorem \cite{AULWAN2}, the observation error dynamics \eqref{dot_tilde_Q_c} is topologically equivalent to an LTV dynamics in a small neighbourhood of these three unstable equilibria. As a result, only some very specific trajectories, from a zero Lebesgue measure set $\calm_\epsilon$, ultimately converge to the unstable equilibria $\bq_i$ ($i=1,2,3$). Combining the local exponential stability of the equilibrium $I_3$, thus it yields the almost global asymptotic stability of the error dynamics \eqref{dot_tilde_Q_c} at $\tilde Q_c = I_3$. It is equivalent to show that
\begequ
\label{tilde_R}
\forall \tilde{Q}_c(0)\in SO(3)\backslash \calm_\epsilon,\quad
\lim_{t\to+\infty}\|\hat Q_c(t) - Q_c\| =0,
\endequ
and $\hat Q_c(t) \in SO(3), \; \forall t\ge 0$. Namely, $\hat R$ provides an asymptotically convergent estimate to $R$ as well.

The last part of the proof is to show the convergence of position estimate $\hat x$. Now define its estimation error as 
$
\tilde x := \hat x - x,
$
and it yields
$$
\begin{aligned}
\dot{\tilde x} & = \hat R\hat v - R v + \sum_{i=1}^{n} \sigma_i ({}^I z_i - \hat x - \hat R \hat z_i)
\\
& = \sum_{i=1}^{n} \sigma_i ({}^I z_i - \hat x - Rz_i) + \et
\\
& = \sum_{i=1}^{n} \sigma_i ({}^I z_i - \hat x - RR^\top ({}^I z_i - x)) + \et,
\end{aligned}
$$
thus 
$$
\dot{\tilde x} =  - \sigma_{\Sigma}\cdot \tilde x + \et
$$
with $ \sigma_{\Sigma}:=\sum_{i=1}^{n} \sigma_i>0$ and an exponentially decaying term $\et$, where in the second equation we have used the convergence \eqref{conv:v_zi} and \eqref{tilde_R}, as well as the compactness of $SO(3)$ and the boundedness of $v$. The dynamics of $\tilde x$ is a linear time-invariant (LTI), stable system perturbed by $\et$, thus also being globally exponentially stable. It completes the proof.
\qed


\section{Simulation details}

In this section, we provide additional details to the simulation results in Section \ref{sec6}.

{\em Position estimation with velocity measurement.} We consider two trajectories: the first trajectory being PE, {\em i.e.}, 
$$
x_1 = \begmat{\cos({t\over 2}) \\ {1\over4} \sin(t) \\- {\sqrt{3}\over 4}\sin(t) }
,\quad
\Omega_1 = \begmat{\sin(0.1+\pi) \\ 0.5\sin(2t) \\ 0.1\sin(0.3t+{\pi\over3})}
$$ 
from the initial attitude $R(0) = I_3$; and the second being IE but not PE with the same scenario as the first one but decaying after $4$s, {\em i.e.},
$$
\begmat{{}^Iv_2 \\ \Omega_2} = \left\{
\begin{aligned}
& \begmat{\dot x_1(t) \\ \Omega_1(t)},  & t\in [0,4]
\\
&  \begmat{e^{-5(t-4)}\dot x_1(t) \\ e^{-5(t-4)}\Omega_1(t) }, & t \ge 4.
\end{aligned}
\right.
$$
The parameters and initial conditions in the observers are selected as
$$
\alpha =1, ~ \gamma = 50,~ \hat r(0) = 0, ~ \xi(0) = 0, ~\hat \theta(0) =0.
$$
For this case, We give some simulation results complementary to Fig. \ref{fig:comparison}. In Fig. \ref{fig:v1} we present the estimation performance of the proposed gradient observer \eqref{grad_obs} and the PEBO \eqref{pebo:1} under two robot trajectories. For the first trajectory, both achieve satisfactory performance; but for the IE trajectory there is a significant ultimate estimation error from the gradient observer. In the feature depth observer \cite{DELetal}, we set $k_i=1$ ($i=1,\ldots,3$) and $\hat x(0) = [0.2~ 0.2~ 0.5]^\top$. In the I\&I observer \cite{KARAST} we select the initial guess as $0.5$ and the gain $\lambda =0.2$. Even though states in these three observers live in different spaces, we tried to make the initial range estimates as close as possible to make a fair comparison. In Figs. \ref{fig:riccati1}-\ref{fig:riccati2}, we also compare to  the LTV Kalman filter \cite{LOUetal} and the Riccati observer \cite{HAMSAM}. The initial conditions of position estimates are both set as $\hat z (0) =0_3$. See also the last figure in Fig. \ref{fig:last}.

\begin{figure*}[!htb]
 \centering
   \subfigure[Range estimate $\hat r$ (PE)]{
   \includegraphics[width=0.27\textwidth]{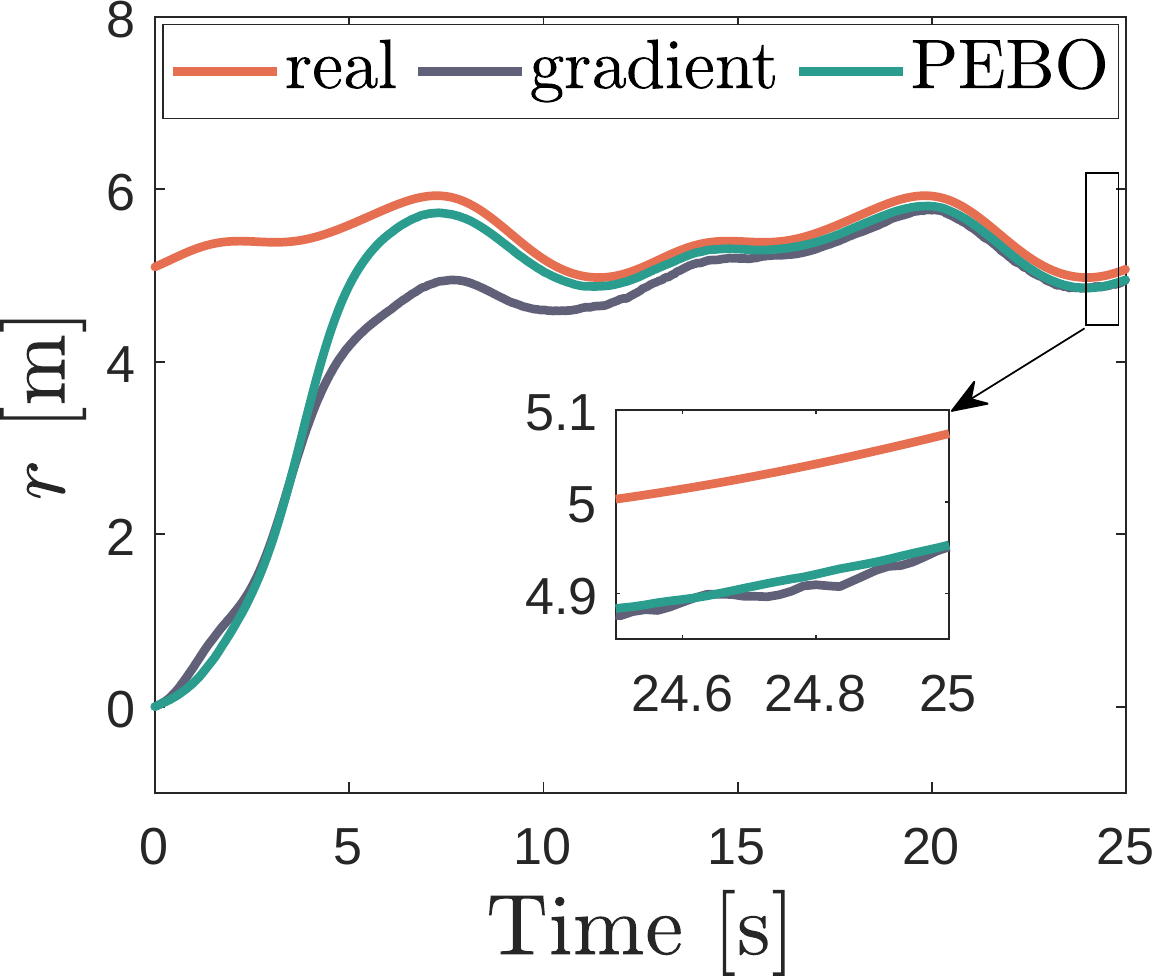}
   \label{fig:with-v1-1}
   }
   \subfigure[Position estimate $\hat z$ (PE)]{
   \includegraphics[width=0.27\textwidth]{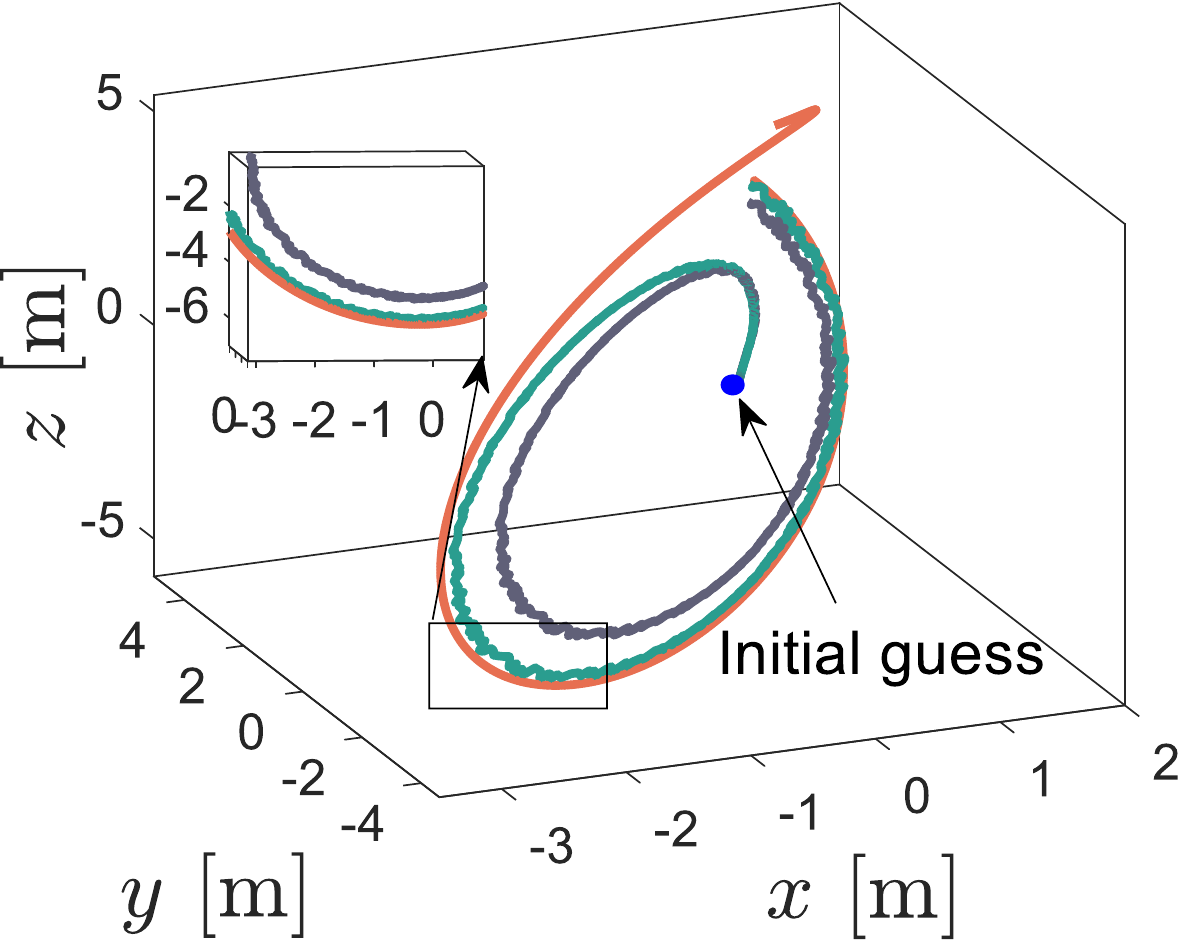}
   \label{fig:with-v1-2}
   }
   \subfigure[Range estimate $\hat r$ (IE)]{
   \includegraphics[width=0.27\textwidth]{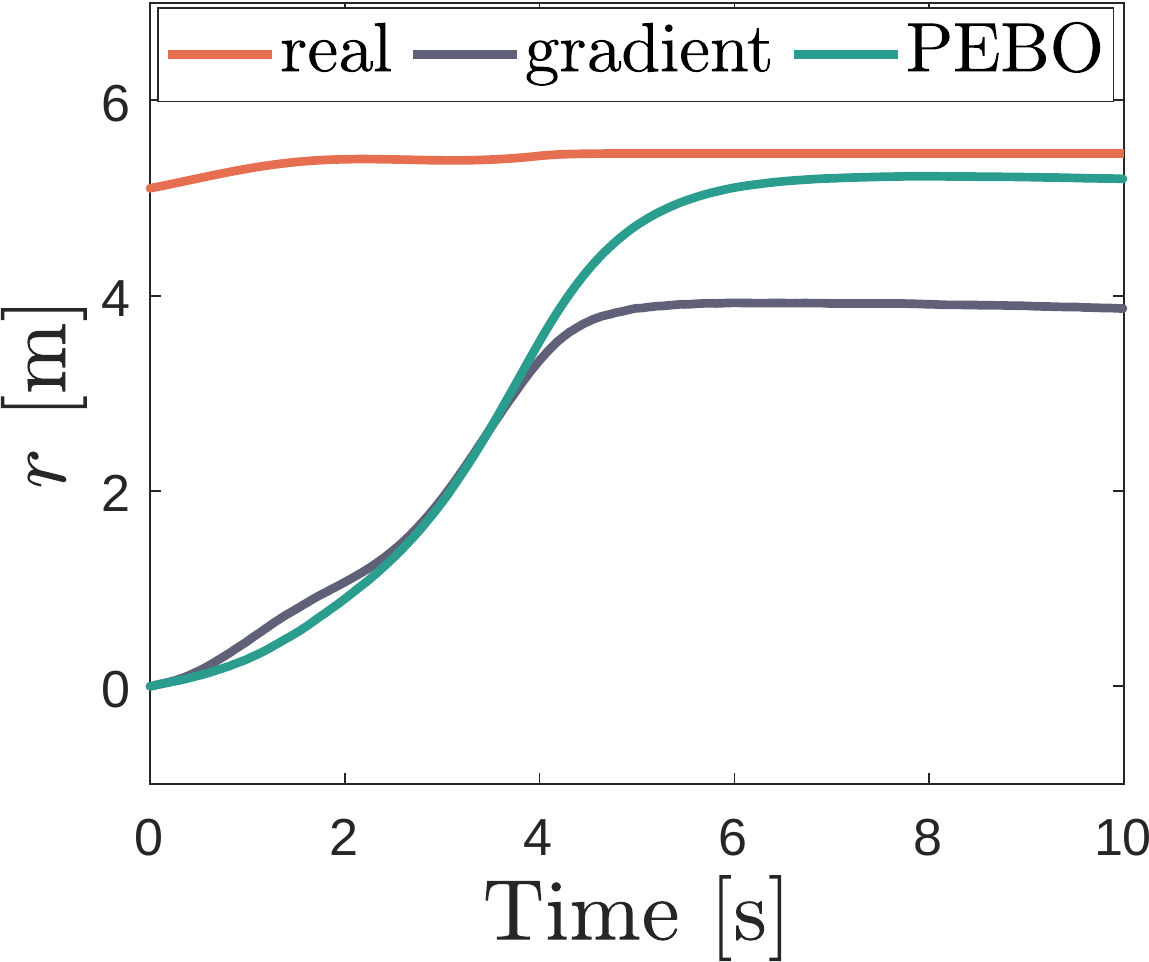}
   \label{fig:with-v1-1}
   }
   \subfigure[Position estimate $\hat z$ (IE)]{
   \includegraphics[width=0.27\textwidth]{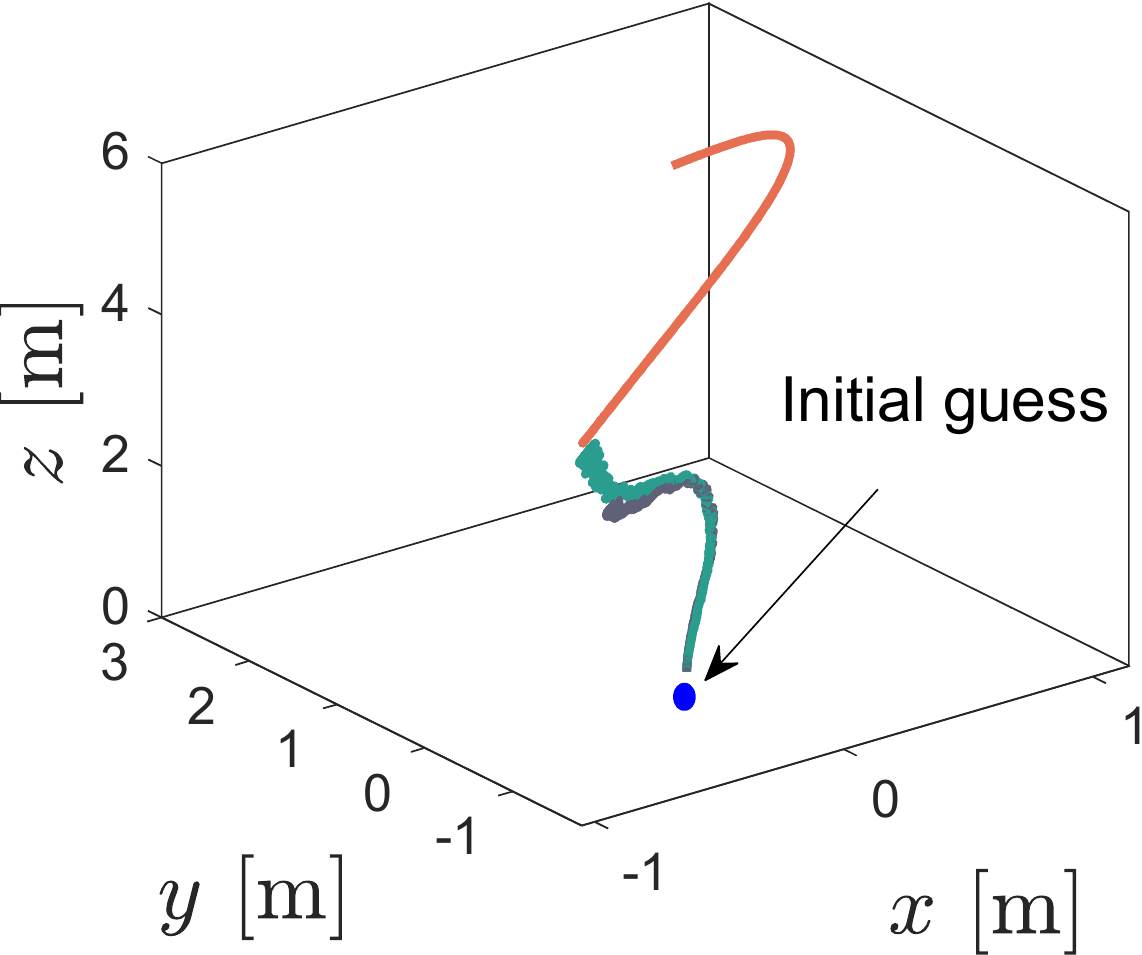}
   \label{fig:with-v1-2}
   }
 \caption{Performance of the proposed position observers under PE and IE trajectories, {\em i.e.}, the gradient observer \eqref{grad_obs} and the PEBO \eqref{pebo:1} for the case with unbiased velocity information and noisy measurement}
 \label{fig:v1} 
\end{figure*}

{\em Position estimation with biased acceleration measurement.}
In the second group of simulations, we consider position and velocity estimation with bearing and biased linear acceleration measurable. The robot trajectory we use is given by for $t\in [0,20]$s
$$
{}^I a  = 
\begmat{-0.5\cos(0.5t) \\ -0.5\sin(t) \\ {\sqrt{3}\over4} \sin(t)}
, \quad
\Omega = 
\begmat{0.2\sin(0.1t+\pi) \\ 0.1\sin(0.2t) \\ 0.1\sin(0.3t+{\pi \over 3})}
$$
and $t\ge 20$s
$$
{}^I a = 0_3, \quad \Omega = 0_3.
$$
Such a trajectory only guarantees the regression being IE rather than PE. The parameters and initial conditions for the observer in Section \ref{sec4} are set as 
$$
\alpha=2,~ \gamma = 100,~ \rho =0.4, ~k_p = 500, ~\hat \theta(0)= [0, \ldots, 0, 10]^\top.
$$
We consider a small sensor bias $b_a = [0.09, 0.10, 0.11]^\top$ and the feature point is located at ${}^I z = [-2,1,3]^\top$. We show in the second sub-figure of Fig. \ref{fig:last} the position $z$ and its estimate $\hat z$ as complement to Fig. \ref{fig:6}.


{\em Vision-aided inertial navigation.}
Finally, we evaluate the performance of the visual inertial navigation observer in Proposition \ref{prop:navigation_observer}. Three known landmarks in the inertial frame are located at $[-2,1,3]^\top, [-2,2,1]^\top$ and $[1,1,1]^\top$. The same robot trajectory and parameters as in the last subsection were adopted, except $\alpha =1$ and $k_p = 10^3$. The behaviour is shown in Fig. \ref{fig:7} and the last sub-figure of Fig. \ref{fig:last} in the absence of measurement noise. Note that, for visualization, in Fig. \ref{fig:last} we draw in the inertial frame $\{\cali\}$ the attitude and its estimate by means of the unit vector along $x$-axis of the body-fixed frame $\{\calb\}$.

\begin{figure*}[!htb]
 \centering
   \subfigure[Position estimate $\hat z$ from PEBO, I\&I observer and feature depth observer (complementary to Fig. \ref{fig:with-v1-10})]{
   \includegraphics[width=0.27\textwidth]{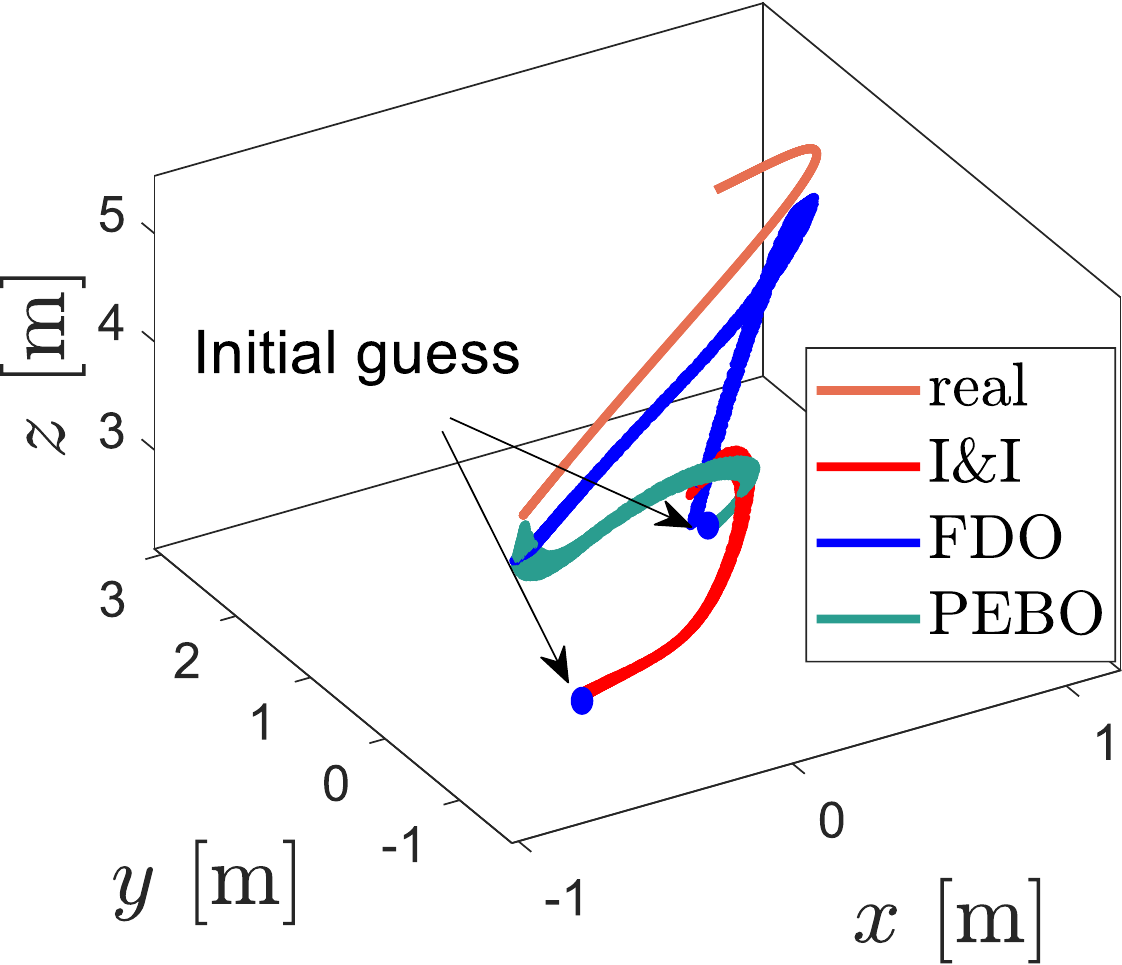}
   \label{fig:with-l1}
   }
   ~
   \subfigure[Position estimate $\hat z$ (complementary to Fig. \ref{fig:a5})]{
   \includegraphics[width=0.27\textwidth]{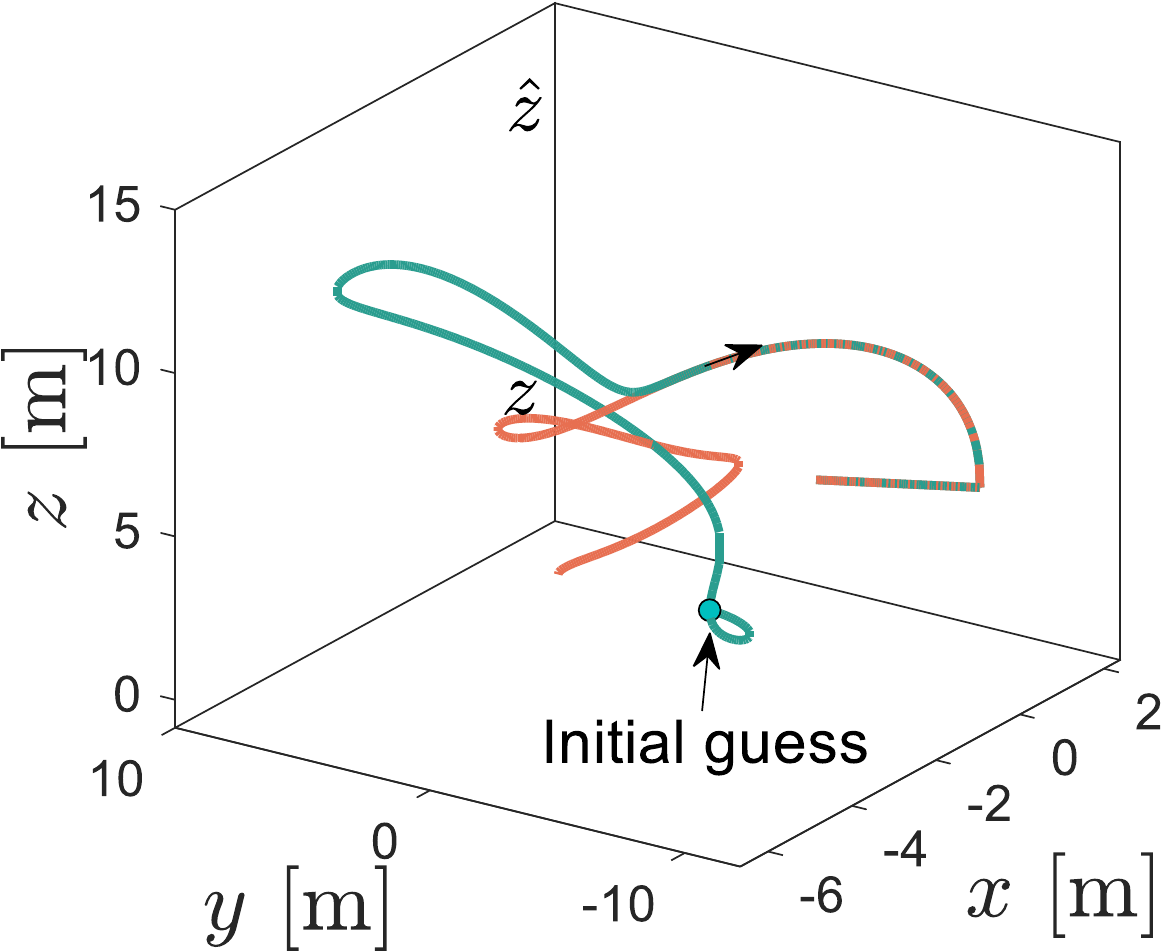}
   \label{fig:with-l2}
   }
     ~
   \subfigure[Attitude estimate $\hat R$ (complementary to Fig. \ref{fig:n1})]{
   \includegraphics[width=0.21\textwidth]{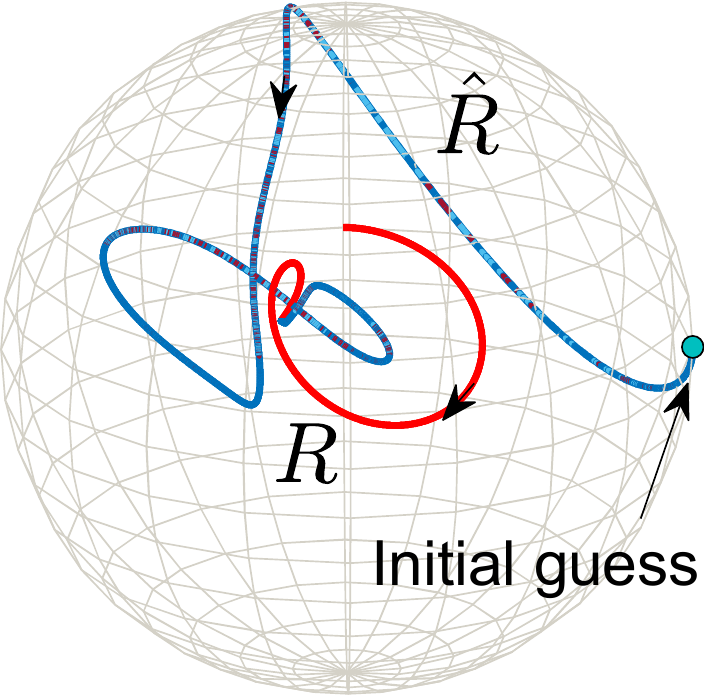}
   \label{fig:with-l3}
   }
 \caption{Some complementary figures: Estimated trajectories and their true values}
 \label{fig:last}
\end{figure*}

\end{document}